\documentclass[11pt]{article}

\usepackage[utf8]{inputenc}
\usepackage{amsmath}
\usepackage{xurl}
\usepackage{amsfonts}
\usepackage{mathtools}
\usepackage{graphicx}
\usepackage{subcaption} 
\usepackage{color}
\usepackage[hidelinks]{hyperref}
\usepackage{tabularx}
\usepackage{float}
\usepackage{array}
\usepackage{multirow}
\usepackage{makecell}
\usepackage{url}
\usepackage{colortbl}
\usepackage{dsfont}
\usepackage{stmaryrd}
\usepackage{natbib}
\usepackage{caption}
\usepackage{color}
\usepackage{xcolor}
\usepackage{algorithm}
\usepackage{algpseudocode}
\usepackage{titlesec}

\renewcommand{\baselinestretch}{1.3}

\makeatletter
\@addtoreset{equation}{section}

\makeatother

\newcounter{Fig}[figure]

\newcounter{Tab}[table]

  {\refstepcounter{Tab}
   \addtocounter{table}{1}
   \samepage\vspace{0.2cm}
   \centerline{#1} \nobreak
   \begin{verse}{Table \thesection.\arabic{Tab}:}
  }%
  {\end{verse} \vspace{0.2cm}}

\relax

\newcommand{\bqa}{\begin{eqnarray*}}
\newcommand{\eqa}{\end{eqnarray*}}
\newcommand{\bqan}{\begin{eqnarray}}
\newcommand{\eqan}{\end{eqnarray}}
\newcommand{\bqt}{\begin{quote}}
\newcommand{\eqt}{\end{quote}}
\newcommand{\bt}{\begin{tabbing}}
\newcommand{\et}{\end{tabbing}}
\newcommand{\bit}{\begin{itemize}}
\newcommand{\eit}{\end{itemize}}
\newcommand{\ben}{\begin{enumerate}}
\newcommand{\een}{\end{enumerate}}
\newcommand{\beq}{\begin{equation}}
\newcommand{\eeq}{\end{equation}}
\newcommand{\bdefi}{\begin{definition}}
\newcommand{\edefi}{\end{definition}}
\newcommand{\bpro}{\begin{proposition}}
\newcommand{\epro}{\end{proposition}}
\newcommand{\blem}{\begin{lemma}}
\newcommand{\elem}{\end{lemma}}
\newcommand{\bth}{\begin{theorem}}
\newcommand{\eth}{\end{theorem}}
\newcommand{\bco}{\begin{corollary}}
\newcommand{\eco}{\end{corollary}}
\newcommand{\bdes}{\begin{description}}
\newcommand{\edes}{\end{description}}
\newcommand{\bre}{\begin{remark}}
\newcommand{\ere}{\end{remark}}

\newtheorem{definition}{Definition}[section]
\newtheorem{proposition}[definition]{Proposition}
\newtheorem{lemma}[definition]{Lemma}
\newtheorem{theorem}[definition]{Theorem}
\newtheorem{corollary}[definition]{Corollary}
\newtheorem{remark}[definition]{Remark}

\newcounter{assump}
\newenvironment{proof}[1][Proof]{\textbf{#1.} }{\ \rule{0.5em}{0.5em}}
\newtheorem{assum}[assump]{Assumption}

\begin{document}

\begin{titlepage}

\title{Combination of traditional and parametric insurance: calibration method based on the optimization of a criterion adapted to heavy tail losses.}

\author{{\large Olivier L\textsc{opez}$^1$},
{\large Daniel N\textsc{kameni}$^{1,2}$}  }

\date{\today}
\maketitle

\renewcommand{\baselinestretch}{1.1}

\begin{abstract}
In this paper, we address the problem of providing insurance protection against heavy-tailed losses, for which the expected loss may not even be finite. The product we study is based on a combination of traditional insurance up to a given limit and a parametric (or index-based) cover for larger losses. This second component of the coverage is computed from covariates available immediately after the loss occurs, allowing claim management costs to be reduced through rapid compensation. To optimize the design of this second component, we use a criterion adapted to extreme losses, that is, to loss distributions of Pareto type. We support the calibration procedure with theoretical results establishing its convergence rate, as well as empirical evidence from both a simulation study and a real-data analysis on tornado losses in the United States. We also propose a two-step optimization procedure as a potential solution to the issue of data scarcity in the tails of loss distributions. We conclude by empirically demonstrating that the proposed hybrid contract outperforms a traditional capped indemnity contract.
\end{abstract}

\vspace*{0.5cm}


\vspace*{0.5cm}

\noindent{\bf Short title:} Combination of Traditional and Parametric insurance for heavy-tailed losses.

\vspace*{.5cm}

{\small
\parindent 0cm
$^1$ CREST Laboratory, CNRS, Groupe des Écoles Nationales d'Économie et Statistique, Ecole Polytechnique, Institut Polytechnique de Paris, 5 avenue Henry Le Chatelier 91120 PALAISEAU, France \\
$^2$ Detralytics, 1-7 Cours Valmy 92923 Paris La Défense Cedex, France. E-mails: olivier.lopez@ensae.fr, daniel.nkameni@ensae.fr}

\end{titlepage}

\small
\normalsize
\addtocounter{page}{1}

\section{Introduction}

The modern landscape of insurance is characterized by the emergence of new risks whose potential of destruction question the feasibility of developing a sustainable financial protection. Several reports\footnote{
see the report for the European Insurance and Occupational Pensions Authority, \url{https://www.eiopa.europa.eu/leveraging-insurance-shore-europes-climate-resilience-2024-09-03_en} or the report on insurability of climate risk mandated by the French ministry of Economy, see \url{https://www.ecologie.gouv.fr/sites/default/files/documents/Rapport_final_Mission-assurance_climat.pdf}.} alerted about the potential uninsurability of natural disasters in the context of climate change. In the context of cyber, the difficulty to cover risks that are increasing in terms of frequency, severity and variety is materialized by significant protection gaps\footnote{See for example the report from the Global Federation of Insurance Associations in March 2023, \url{https://gfiainsurance.org/topics/487}.}. To overcome these challenges, parametric insurance (or "index-based insurance") is often mentioned as one of the technical tools that may help offering a coverage in these apparently desperate situations: a report from the French Direction of Treasury on cyber insurance\footnote{see \url{https://www.tresor.economie.gouv.fr/Articles/2022/09/07/remise-du-rapport-sur-le-developpement-de-l-assurance-du-risque-cyber}.} mention these "innovative products" as a potential solution to respond to the need for financial protection against this expanding threat.

The concept of parametric insurance offers indeed a convenient solution to the difficulty to cover risks that are difficult to apprehend, either because they are too volatile, or because lack of data makes it hard to develop models required to build an efficient and profitable protection. The trick is essentially not to cover the risk itself, but to build a product for which compensation is determined by the values taken by a parameter (or index) available just after the claim. Multiple examples of application can be found for example in \cite{abdi2022index}, \cite{eze2020exploring}, \cite{lin2020application}, \cite{carter2017index} or \cite{alderman2007insurance}. Of course, the compensation mechanism has to be determined so that it matches with the true loss experienced by the policyholder, that is one needs to control the "basis risk" (see \cite{doherty2002moral} or \cite{elabed2013managing}). But if this difference between the loss and the compensation is contained, the product benefits from many advantages:
\begin{itemize}
\item the parameter is usually considered because many data are available on it and/or because its distribution is easier to control from a risk management perspective than the one of the true loss. Better knowledge of the parameter means more control on the compensation mechanism.
\item the availability of the parameter reduces expertise and claim management costs, it also ensures a faster compensation for the policyholder.
\item the rules are clear, and this should lead to fewer legal dispute between policyholders and insurers.
\end{itemize}

Another advantage of this type of mechanism is the possibility to develop Insurance Linked Securities that are based on these parametric products, allowing to transfer the risk to capital markets, see \cite{braun2023cyber}, \cite{securities2009handbook}, \cite{michel2008extreme}. Although we will not explore this dimension in the present paper, focusing on the ability of such a product to cover a damage, this dimension should not be forgotten as an important leverage to mobilize more capital (via the markets) in the absorption of catastrophic risks, making these products more appealing in a broader perspective.

In the present paper, we consider the design of parametric insurance as a complement of traditional insurance, when the distribution of the loss is heavy-tailed. In the field on Insurance Linked Securities, \cite{hardle2010calibrating} or \cite{mastroeni2022pricing} are examples where a Pareto-type variable models the underlying risk. In this case, full coverage of the loss by a traditional insurance contract may be difficult due to a too important premium, or even impossible if the expectation of the loss is infinite. The parametric insurance product is used on top of a limited insurance contract to improve the protection at an affordable price. Many market studies show the popularity of parametric solutions to cover the upper segment of the risk\footnote{see for example the report from PricewaterhouseCoopers, \url{https://www.pwc.ch/en/insights/fs/closing-the-gap-with-parametrics-insurance.html}}.

\textbf{Review of literature}.

Relying on index based products (parametric insurance or catastrophic bonds) has been widely considered to deal with various types of risks. This is for example the case in \cite{BIAGINI2008214} and \cite{gatzert2019convergence} for catastrophic risk. Agriculture insurance is also a major domain in which parametric products have been used, see for example \cite{ozaki2008parametric} or \cite{prokopchuk2020parametric}. Index based products have also been considered in pandemic risk management (see \cite{huang2019estimating}, with sometimes mixed experiences as in \cite{jonas2019pandemic}) and in life insurance when it comes to deal with longevity risk, see for example \cite{Cairns18022021}. The question of dealing with basis risk, an important limitation of parametric insurance, is also well documented, as in \cite{chantarat2013designing}, \cite{clarke2016theory} or \cite{jensen2016index}.

To design optimal parametric insurance, several approaches have been studied in the literature. A first category includes optimization of criteria related to utility theory, as in \cite{louaas2021optimal}, or \cite{tan2024flexible}, with various elaborate techniques have been proposed, from regression trees to neural networks in \cite{chen2024managing}. In presence of heavy tail risk, \cite{conradt2015tailored} relied on quantile regression, see also \cite{GU2023106350} in the case of earthquakes. 

Our approach differs in that we consider a different metric — one that focuses on the ratio between the compensation offered by the product and the loss suffered, rather than on their difference. This framework seems adapted to the study of heavy tail distributions, first of all because it avoids dealing with quantities that may have infinite expectation. Also because, in this context, the risk is so important that it is implicit that it is impossible to cover it perfectly. On the other hand, the percentage of the loss that is left to the policyholder is a concrete indicator of the performance of the solution, on which it is easy to communicate. Moreover, our approach combines indemnity based system (based on the true value of the loss) and a index-based component.

Recently, a similar combination of traditional insurance and a parametric product designed via neural networks has been proposed by \cite{zhang2024blended}, but in an opposite configuration, where the parametric product is here to deal with small claims, therefore with no need to confront to the heavy tail nature of the risk in the design of the index.

\textbf{Structure of the paper.}

The contributions of this paper are driven by the need to determine the optimal protection mechanism of the policyholder, when this mechanism is made of a combination of traditional insurance, supplemented by some parametric product to cover the upper segment. Our analysis is performed in four steps:
\begin{enumerate}
\item We first define in section \ref{sec:util} a function that we want to optimize, in conceiving the product. The idea is to find a compromise between the proportion of the loss that is covered, and an affordable price.
\item Once this metric has been defined, we define an optimal cover based on the available information after the claim (see section \ref{sec:risk} to \ref{sec:optim}).
\item We provide a statistical methodology to approach this optimal policy based on available data (section \ref{sec:optimization}).
\item We investigate the behavior of the procedure on simulations (section \ref{sec:simul}) and illustrate the feasibility of the method, on tornado loss data from the US (section \ref{sec:real}).
\item We compare the proposed contract to the capped indemnity contract which is generally used by insurers (section \ref{sec:comp_analysis})
\end{enumerate}

The technical proofs are postponed to section \ref{sec:appendix}.

\section{Index based reinsurance}
\label{sec:util}

This section is devoted to defining the framework of the insurance cover that we study in this paper, namely a combination of "traditional" and parametric insurance. The first step is to introduce, in section \ref{sec:putility}, a proper metric, adapted to measure what is expected from this product, in the context of heavy tail losses. A part of this metric is made to a way to materialize the aversion to high prices, which is discussed in section \ref{sec:choicef}. We then describe a specific type of product based on a combination of traditional insurance with parametric insurance taking over after some threshold in section \ref{sec:risk}. A first way to design the index-based part of the product is discussed in section \ref{sec:optim}.

\subsection{A metric to optimize}
\label{sec:putility}

In the following, we consider a policyholder who will be affected by the loss $Y$ (which may be zero if no claims happen during the period under study). In this section, we propose a way to measure how an insurance contract responds to the needs of the policyholder, in the particular case where $Y$ is heavy-tailed, that is

\begin{equation}\label{eq:type}
S_Y(t)=\mathbb{P}(Y > t) = \frac{l(t)}{t^{1/\gamma}},
\end{equation}

where $\gamma>0$ is called the tail index and $l$ is a slow-varying function, that is
\begin{equation}\forall h>0, \lim_{t\rightarrow \infty}\frac{l(ht)}{l(t)}= 1. \label{slowvarying}\end{equation} This class of distributions plays an important role in extreme value analysis, see \cite{beirlant}, and is particularly adapted to consider risk with huge potential losses.

Consider an insurance contract that offers the compensation $X,$ usually smaller than $Y.$ We want to maximize the proportion of the risk that is covered, that is the ratio $X/Y.$ In case of large risks, this indicator can be more convenient than looking at the difference between $X$ and $Y:$ if the tail of $Y$ is heavier than the tail of $X,$ this may exhibit large gaps that blur the perception of the performance. The ratio is less affected by this phenomenon. Hence, in the following, we consider that the aim is to maximize
\begin{equation}\label{eq:general}\mathfrak{L}_{X}=E\left[L\left(\frac{X}{Y}-f(\pi_X)\right)\right],
\end{equation}
where $L$ is a non-decreasing function, $f$ is a non-decreasing function that materializes the aversion to high prices, and $\pi_X$ is the loaded insurance premium associated with the compensation $X$. We adopt the convention $0/0=1.$

For the function $L,$ an obvious choice is the identity $L(x)=x$, but other choices could be more suitable as discussed in section \ref{sec:approxim}. The role of the function $f$ is discussed in section \ref{sec:choicef}.

\subsection{Price aversion}
\label{sec:choicef}

We now discuss the choice of the function $f.$ In (\ref{eq:general}), the benefits of an increase of the ratio $X/Y$ can be compensated by an increase of the price, hence of $f(\pi_X).$ Since the ratio $X/Y$ is supposed to be bounded\footnote{The spirit is to have $X/Y\leq 1.$ Overcompensation can exist, but should be particularly rare, otherwise the product would not be interesting anymore for the insurer. Therefore $X/Y$ is expected to be bounded by a constant slightly larger than 1 except of an event a small probability.}, it is natural to consider a non-decreasing function $f:\mathbb{R} \rightarrow 
 [0,1]$ that is bounded, like, for example,
$$f_l(\pi)=\frac{\kappa}{1+\exp(-\beta \pi)},$$
and a rational fraction
$$f_r(\pi)=\frac{\kappa \pi^{\beta}}{1+\pi^{\beta}}.$$

From a statistical perspective, the function $f$ plays the role of a penalty function in the selection of the best cover $X$ from a given family. In this view, the quantity $\kappa$ in $f_r$ and $f_l$ can be understood as a hyperparameter that controls the force of the penalization.

Let us discuss some implications of the choice of a specific form of $f.$ Consider the special case where, $X=X(s)=\min(Y,s).$ 
We have
$$\pi(s)=E[\min(Y,s)]=\int_0^s S_Y(t)dt.$$
Moreover, the function (\ref{eq:general}) (that we here denote $\mathfrak{L}(s)$ to shorten the notation) becomes
\begin{equation}
\label{eq:partial}
\mathfrak{L}(s)=(1-S_Y(s))L\left(1-f(\pi(s))\right)-\int_s^{\infty}L\left(\frac{s}{t}-f(\pi(s))\right)dS_Y(t).
\end{equation}
The derivative of (\ref{eq:partial}) with respect to $s$ is
\begin{equation}
\nonumber \mathfrak{L}'(s)=-\pi'(s)f'(\pi(s))(1-S_Y(s))L'\left(1-f(\pi(s))\right)-\int_s^{\infty} L'\left(\frac{s}{t}-f(\pi(s))\right)\left[\frac{1}{t}-\pi'(s)f'(\pi(s))\right]dS_Y(t).
\end{equation}

Let us consider a few particular situations:
\begin{itemize}
\item[1.] $E[Y]=\pi^+<\infty.$ When $s$ tends to infinity, $\pi(s)\rightarrow \pi^+$ and 
$\pi'(s)=S(s)\rightarrow 0.$ If $f'(\pi^+)>0,$ the first term is $$-\frac{l(s)f'(\pi^+)}{s^{1/\gamma}}L'\left(1-f(\pi^+)\right)+o\left(\frac{1}{s^{1/\gamma}}\right),$$
while the second is positive but is $O(s^{-1/\gamma-1}).$ Hence, for $s$ large enough, $\mathfrak{L}'(s)<0,$ which means that it corresponds to a situation where the policyholder is not necessarily seeking for a perfect protection against the risk (which would be $s=\infty$): after some point, the burden of a high premium is too important, and $\mathfrak{L}$ decreases.
\item[2.] $E[Y]=\infty.$ Then, if $f(\pi(s))$ is bounded, $f'(\pi(s))\rightarrow 0$ as $s$ tends to infinity. In this case, $\mathfrak{L}'(s)$ is equivalent, for large $s,$ to $-l(s)f'(\pi(s))s^{-1/\gamma}L'(1-f(\pi(s)))$ only if it gets to zero faster than $s^{-1/\gamma-1}$ (the positive term in the decomposition of $\mathfrak{L}'$), which happens if $l(s)sf'(\pi(s))\rightarrow \infty.$ If this condition does not hold, this corresponds to the situation where the policyholders are ready to pay any amount of premium against a perfect protection. Note that condition $l(s)sf'(\pi(s))\rightarrow \infty$ holds for function $f_r$ if $\beta<\gamma/(\gamma-1)$ in the case $\gamma>1,$ since $\pi(s)=O(s^{1-1/\gamma})$ and $f_r(\pi)=O(\pi^{-(1+\beta)})$ when $\pi$ tends to infinity. On the other hand, it is not the case for $f_l$ since $f_l(\pi)=O(\exp(-\beta \pi)).$
\end{itemize}

Clearly, the practical choice of $f$ should be calibrated from data describing the behavior of the policyholders, but this discussion shows that some shapes can be eliminated based on some basic knowledge on their behavior regarding the possibility of full protection (ready to pay any price or not).

We now present the general framework (\ref{eq:general}) in the particular case of the parametric insurance protection that we aim to study.

\subsection{Risk and index based compensation}
\label{sec:risk}

In the following, we consider a loss variable $Y$ with survival function $S_Y(t)=\mathbb{P}(Y\geq t).$ The policyholder wants protection against this loss for an affordable price. On the other hand, the insurance coverage is partial, and the pay-off of the insurance contract is
\begin{equation}X_{\phi}(s)=Y\mathbf{1}_{T(\mathbf{W})\leq s}+s\phi(\mathbf{W})\mathbf{1}_{T(\mathbf{W})>s}, \label{eq:xphi}
\end{equation}
with $\mathbf{W}\in \mathbb{R}^d$ is a set of variables that are observed soon after the claim, and will allow to compute a trigger mechanism $T(\mathbf{W})$ and a compensation $s\phi(\mathbf{W})$. The spirit is that the event $\{T(\mathbf{W})>s\}$ should approximately correspond to the event $\{Y>s\}$. In our framework, the parametric insurance product is present to help dealing with large claims that would be beyond insurance capacity. On the other hand, we do not want to rely on precise evaluation of the loss $Y,$ which would be in contradiction with the logic behind parametric insurance (which consists in avoiding relying on expert analysis). Since we expect the parametric part to activate when $Y>s,$ $\phi(\mathbf{W})$ should, in most cases, be larger than 1, but it could be authorized to have $\phi(\mathbf{W})$ less than 1 to "smooth" the behavior around $s$ and to avoid over-compensation . Moreover, the resulting under-compensation could be tolerated by the policyholder due to the fast payment associated with index-based products. The threshold $s$ is essentially a way for the insurer to protect himself against the too large values taken by $Y.$ But instead of simply giving a limit of compensation $s,$ once this threshold is reached, a new mechanism of compensation is used, based on $\mathbf{W}.$

This framework is designed to mimic the behavior of a parametric (or "index-based") insurance products that completes the traditional cover. This logic is similar to \cite{zhang2024blended}, but in a context where the index-based component is activated for large claims instead of small ones. This corresponds to the use of parametric insurance to improve insurability in cases where the market considers that traditional insurance is not profitable enough beyond a given capacity, while \cite{zhang2024blended} corresponds to the other important situation where the administrative costs of small claims are reduced through this channel. The behavior of parametric insurance products can be considered in two steps:
\begin{itemize}
\item a triggering mechanism, namely a condition for the compensation to activate;
\item a compensation, if activated.
\end{itemize}

Regarding the trigger mechanism $T(\mathbf{W}),$ many options are available via binary regression (from logistic regression, quantile regression, to machine learning methods). Indeed, $T(\mathbf{W})$ is typically a score used to predict if $Y>s$ or not. We do not focus on the specific way to build this trigger, since this binary problem is less complex than predicting the value of the heavy tail variable $Y$ itself. We will only state generic conditions on $T(\mathbf{W})$ when it comes to providing our theoretical results.

We can see the product $X_{\phi}$ as a kind of reinsurance mechanism, where $s$ would represent the priority of a stop-loss reinsurance treaty. The quantity $s\phi(\mathbf{W})$ is expected to be close to $Y,$ and, essentially, $s\phi(\mathbf{W})\leq Y$ (at least with high probability) otherwise the compensation would exceed the true loss (a financial product could be designed with such pay-off, but this is not the spirit of insurance products which only aim at repairing and not at generating profit for the policyholder). The pure premium corresponding to this contract is denoted
$\pi^*_{\phi}(s)=E\left[X_{\phi}(s)\right].$ It is assumed to be well defined, that is $\pi_\phi^*(s)<\infty,$ meaning that $E[X]<\infty$ even in the situation where $E[Y]=\infty.$ Moreover, qualitatively, we should keep in mind that this quantity is less than the pure premium (if defined) of an insurance contract entirely covering $Y,$ otherwise this insurance mechanism lacks interest. The final premium $\pi_{\phi}(s)$ paid by the policyholder includes some loading factor $\tau,$ that is
$$\pi_{\phi}(s)=(1+\tau)\pi^*_{\phi}(s).$$
Let us note that here and in the rest of the paper, the frequency of occurrence of the claims is taken into account implicitly: the variable $Y$ may take the value 0 with a non zero probability, materializing the fact that a proportion of policyholders will not experience any claim.

\subsection{A first method of selection of an optimal payoff function $\phi$}
\label{sec:optim}

The question is now to find the appropriate payoff function $\phi$ to maximize (\ref{eq:general}). This is not solely a matter of approximating the true loss $Y:$ if it was the case, and if $E[Y^2]$ were finite, a possible way to proceed would be to consider $\phi(\mathbf{w})=s^{-1}E[Y|Y\geq s, \mathbf{W}=\mathbf{w}],$ which would minimize the mean square error between $X_{\phi}(s)$ and $Y.$ Here, quadratic loss is not the proper metric for heavy tail variables (that may even have infinite expectation and/or variance). Moreover, this would not respect the aversion to high prices materialized by the function $\phi.$

To select the proper function $\phi,$ let us consider that we have at our disposal a class of potential functions $\mathcal{F}=\{\phi_{\theta};\theta \in \Theta\},$ with $\Theta\subset\mathbb{R}^k.$ The dimension $k$ may be very large, for example if one relies on an over-parametrized machine learning model.

In \cite{louaas2021optimal} and \cite{chen2024managing}, the authors propose to maximize an empirical version of the  objective function to select the appropriate value of $\theta$ (hence the corresponding function $\phi_{\theta}$). Transposed in our context, and assuming that we have i.i.d. historical data $(Y_i,\mathbf{W}_i)_{1\leq i\leq n},$ the idea would be to maximize over $\Theta$ the quantity, leading to
\begin{equation}
\label{eq:thetan}\hat{\tilde{\theta}}_n(s)=\arg\max_{\theta\in \Theta} \frac{1}{n}\sum_{i=1}^n L\left(\frac{X_{\theta,i}(s)}{Y_i}-f(\hat{\pi}_{\theta}(s))\right),
\end{equation}
where, to shorten the notation, we use $X_{\theta,i}(s)$ (resp. $\pi_{\theta}(s)$) instead of $X_{\phi_{\theta},i}(s)=Y_i\mathbf{1}_{Y_i\leq s}+s\phi_{\theta}(\mathbf{W}_i)\mathbf{1}_{Y_i>s}$ (resp. $\pi_{X_{\phi_{\theta}}(s)}$), $\hat{\pi}_{\theta}(s)=(1+\tau)n^{-1}\sum_{i=1}X_{\theta,i}(s).$

Let us note that, if the dimension of $\Theta$ is small compared to $n,$ which will be the case in a classical parametric model, under standard assumptions for M-estimators (see \cite{van1996weak}), we get $\hat{\tilde{\theta}}_n(s)-\tilde{\theta}(s)=O_P(n^{-1/2}),$ where
\begin{equation}
\label{eq:theta*}\tilde{\theta}(s)=\arg \max_{\theta\in \Theta}E\left[L\left(\frac{X_{\theta}(s)}{Y}-f(\pi_{\theta}(s))\right)\right]
\end{equation}
corresponds to the parameter of the best payoff function from the class $\mathcal{F}.$ This rate of convergence is informative on the volume of historical data required to approximate the ideal payoff. One of the practical difficulties comes from the fact that the size of this database can be limited, especially for emerging risks: this may be hard to collect a large amount of data where $Y_i$ and $\mathbf{W}_i$ are simultaneously present. On the other hand, it is much easier to get additional data on $\mathbf{W},$ since the measure of these variables used to compute the index can be obtained without any expansive and/or time-consuming expertise.

Therefore, based on a closer analysis of (\ref{eq:general}), we will recommend a two-step procedures that can, in some cases, increase the convergence rate, as described in section \ref{sec:optimization} below.

\section{A new metric to optimize an index-based product}\label{sec:optimization}

In the section, we consider a collection of covers
$$X_{\theta}(s)=Y\mathbf{1}_{T(\mathbf{W})\leq s}+s\phi_{\theta}(\mathbf{W})\mathbf{1}_{T(\mathbf{W})>s},$$
where we want to select, from data, the best function $\phi_{\theta}$ from a collection $\mathcal{F},$ using the notations of section \ref{sec:optim}.

Our first task is to obtain an approximation of $$\mathfrak{L}_\theta(s)=E\left[L\left(\frac{X_{\theta}(s)}{Y}-f(\pi_{\theta}(s))\right)\right],$$ which is done in section \ref{sec:approxim}. From this approximation, we deduce in section \ref{sec:alternative} an alternative way to determine an estimate of the ideal cover, that is the quantity $\tilde{\theta}(s)$ defined by (\ref{eq:theta*}). The benefits of relying on this alternative to the approximation $\hat{\tilde{\theta}}(s)$ defined in (\ref{eq:thetan}) are shown in section \ref{sec:crate} where we study the convergence rates of the different procedures.

\subsection{Approximation of the criterion}
\label{sec:approxim}

Our first step is to investigate the behavior of $\mathfrak{L}_{\theta}(s)$ when $s$ is large. The need for this approximation comes from the fact that we want to design an index-based cover which completes the coverage of the policyholder only for large claims. The design of this parametric part of the product depends on analyzing the conditional distribution of $Y$ given $\mathbf{W}$ where we recall that $\mathbf{W}$ are covariates that are available just after the claim. Modeling the whole distribution of $Y|\mathbf{W}$ can be relatively complex, and the idea is that our final result should be essentially influenced by the tail distribution. Getting rid of the remaining part of the distribution is a simplification, avoiding to make too strong assumptions. This is the spirit of (\ref{eq:type}), where the distribution of $Y$ is specified only in terms of rate of decay of $S_Y,$ without specifying completely the function $l.$

Before stating Theorem \ref{th_approx}, which is the main result of this section, we need some generic assumption on the distribution of $Y|\mathbf{W}$ and the functions $L$ and $f$ involved in the computation of $\mathfrak{L}_{\theta}(s).$

\begin{assum}
\label{a-1}
Let $S(t|\mathbf{w})=\mathbb{P}(Y\geq t|\mathbf{W}=\mathbf{w}).$
Assume that, for all $w,$
$$S(t|\mathbf{w})=\frac{l(t|\mathbf{w})}{t^{1/\gamma(\mathbf{w})}},$$
where $\gamma(\mathbf{w})>0,$ and $l(\cdot|\mathbf{w})$ is a slow-varying function. Moreover, assume that, for some constants $C$ and $c>0$, there exists a slow varying function $l^{0}$ such that

$$cl^{0}(t)\leq l(t|\mathbf{w})\leq Cl^{0}(t).$$

Moreover we assume that
$$\sup_{\theta,\mathbf{w}}\left|\frac{l(st|\mathbf{w})}{l(s|\mathbf{w})}\right|\rightarrow 1,$$
when $s$ tends to infinity.
\end{assum}

Assumption \ref{a-1} above states that the conditional distribution of $Y$ remains in the same family (\ref{eq:type}), but with index parameter potentially depending on $\mathbf{W}.$ The introduction of $\mathbf{l}$ is essentially here to allow us to apply the dominated convergence theorem, controlling the fact that there is no structural change in the conditional distribution of $Y.$

Next, we need two assumptions on the function $L$ and $f.$

\begin{assum}
\label{a0} Let $\pi^+=\lim_{s\rightarrow \infty}\pi_{\theta}(s).$
There exists a function $\varphi_0(t)> 0$ such that
$$\sup_{t}\left|\frac{\frac{L\left(t-f(\pi)\right)}{L\left(1-f(\pi)\right)}-\varphi_0\left(t\right)}{\varphi_0(t)}\right|\rightarrow_{\pi\rightarrow \pi^+}0.$$
\end{assum}

\begin{assum}
\label{a1}
There exists a function $\varphi_1(t)< 0$ such that
$$\sup_{t}\left|\frac{\frac{L'\left(t-f(\pi)\right)}{L\left(1-f(\pi)\right)}-\varphi_1\left(t\right)}{\varphi_1(t)}\right|\rightarrow_{\pi\rightarrow \pi^+}0.$$
\end{assum}

Regarding the definition of the function $L,$

\begin{enumerate}
\item If $L(x) = L_1(x)=x,$ then
\begin{eqnarray*}
\varphi_0(t) &=& \frac{t-f(\pi^+)}{1-f(\pi^+)}, \\
\varphi_1(t) &=& \frac{1}{1-f(\pi^+)},
\end{eqnarray*}
under the condition that $f(\pi^+)<1.$

Another definition of $L$ could be : 
\item $L_2(x)=-\exp(-\mu x)$. In this case, 
\begin{eqnarray*}
\varphi_0(t) &=& \exp(-\mu(t-1)),\\
\varphi_1(t) &=& -\mu \exp(-\mu(t-1)).
\end{eqnarray*}
\end{enumerate}

$L_2$ introduces non-linearity and concavity into the resulting optimization problem. This has the advantage of stabilizing the optimization and numerical results, and also facilitates the identification of the global optimum. Moreover, $L_2$ is more robust to outliers and large values—an especially useful property when dealing with extreme losses.

These advantages of $L_2$ over $L_1$ can be explained by the fact that $L_2$ leads to an optimization that implicitly models the trade-off between mean and variability, unlike $L_1$, which only targets the mean. Indeed, if $X_\theta$ is a random variable, applying a Taylor expansion around $\mathbf{E}[X_\theta]$ gives:

\begin{equation*}
    \mathbf{E}[-\exp(-\mu X_\theta)] \approx -\exp\left(-\mu \mathbf{E}[X_\theta] + \log\left\{1 + \frac{\mu^2}{2}\mathbf{V}[X_\theta]\right\}\right)
\end{equation*}

Thus, $L_2$ approximates a penalty on the variance.

Finally, through a Taylor expansion, $L_2$ can also be interpreted as a generalization of $L_1$ that incorporates higher-order powers of the variable of interest.

\begin{assum}
\label{assum_pdelta}
Assume that there exists a function $\Phi$ such that, for all $\theta$ and $\mathbf{w},$
$$\phi_{\theta}(\mathbf{w})\leq \Phi(\mathbf{w}).$$
Let
\begin{eqnarray*}
E_{-}(s) &=& \left\{Y\leq s \; \text{and} \; T(\mathbf{W})\geq s\right\} , \\
E_{+}(s) &=&  \left\{Y>s \; \text{and} \; T(\mathbf{W})\leq s\right\},
\end{eqnarray*}
$F(s)=E_-(s)\cup E_+(s),$ and
$$p(s)=E\left[\max\left(\frac{s\Phi(\mathbf{w})}{Y}-1,1\right)\mathbf{1}_{F(s)}\right].$$
Assume that $\lim_{s\rightarrow \infty}p(s)=0.$
\end{assum}

Assumption \ref{assum_pdelta} is essentially here to reflect the quality of the trigger procedure. It is required to ensure that the sets $\{T(\mathbf{W})\geq s\}$ and $\{Y\geq s\}$ are approximately the same. Note that, in the case where $s\Phi(w)$ tends to stay much lower than $Y$ (which will usually be the case: the parametric product is expected to have lighter tail than $Y$ to facilitate risk management), the quantity $p(s)$ is close to $\mathbb{P}(F(s)),$ which somehow corresponds to "trigger failures."

\begin{theorem}
\label{th_approx}

Assume that $\sup_x |L'(x)|<\infty.$ Let $$\mathfrak{L}^*_{\theta}(s)=L(1-f(\pi_{\theta}(s)))E\left[1-S(s|\mathbf{W})\Phi_0(\phi_{\theta}(\mathbf{W}),\gamma(\mathbf{W}))
\right],$$
where
\begin{eqnarray*}\Phi_1(x,\gamma) &=& \int_0^x v^{1/\gamma}\varphi_1(v)dv, \\
\Phi_0(x,\gamma)&=&1-\varphi_0(x)+\Phi_1(x,\gamma).
\end{eqnarray*}
Under Assumption \ref{a-1} to \ref{assum_pdelta},
$$\mathfrak{L}_{\theta}(s)=\mathfrak{L}_{\theta}^*(s)(1+\mathfrak{R}_{\theta}(s))+O(p(s)),$$
with $\sup_{\theta}|\frak{R}_{\theta}(s)|\rightarrow 0$ when $s$ tends to infinity.

\end{theorem}

The proof is postponed to section \ref{sec:proof1}. By ensuring that $\mathfrak{L}_{\theta}(s)$ is asymptotically equivalent to $\mathfrak{L}^*_{\theta}(s),$ the idea is that optimizing $\mathfrak{L}^*_{\theta}(s)$ should approximately lead to the same pay-off function.
 One can note that 
\begin{equation}\mathfrak{L}^*_{\theta}(s)=E\left[\Psi(\pi_{\theta}(s),\phi_{\theta}(\mathbf{W});S(s|\mathbf{W}),\gamma(\mathbf{W}))\right].\label{Epsi}\end{equation}

where 
$$\Psi(\pi_{\theta}(s),\phi_{\theta}(\mathbf{W});S(s|\mathbf{W}),\gamma(\mathbf{W})) = L(1-f(\pi_{\theta}(s)))\left[1-S(s|\mathbf{W})\Phi_0(\phi_{\theta}(\mathbf{W}),\gamma(\mathbf{W}))
\right]$$

One can also observe, in $\mathfrak{L}^*_{\theta}(s),$ a separation between a term $L(1-f(\pi_{\theta}(s)))$ that materializes the effect of the price, and the term between brackets. The two terms drive $\mathfrak{L}$ in two opposite directions.

Another interesting feature of this decomposition is that we see that the whole joint distribution of $Y$ and $\mathbf{W}$ is not involved in $\Psi$: to compute the approximation of $\mathfrak{L}(s),$ one only needs to compute the conditional tail index $\mathbf{w}\rightarrow \gamma(\mathbf{w})$  and $\mathbf{w}\rightarrow S(s|\mathbf{w}),$ and to know the distribution of $\mathbf{W}.$ Essentially, one needs to estimate the probability of triggering the parametric compensation $s\phi_{\theta}(\mathbf{w})$ (which is easier to do than if we needed to estimate the whole conditional distribution of $Y$ above $s$), and since the information $\mathbf{W}$ is usually more available than precise claim data on $Y,$ this distribution can be estimated with a good accuracy (see section \ref{sec:alternative} below).

\subsection{Alternative optimization for the parametric part}\label{sec:alternative}

Let us consider that we have two databases. The first contains joint i.i.d. replications $(Y_i,\mathbf{W}_i)_{1\leq i \leq n}$ of $(Y,\mathbf{W}).$ The second contains i.i.d. replications $(\mathbf{W}_j)_{1\leq j \leq m}$ of $\mathbf{W}$ (the first $n$ observations are common to the two databases, although there is no obligation).

Theorem \ref{th_approx} encourages to rely on a two step procedure to optimize the index based cover:
\begin{itemize}
\item Step 1: from $(Y_i,\mathbf{W}_i)_{1\leq  i \leq n},$ estimate $w\rightarrow S(s|\mathbf{w})$ and $\mathbf{w}\rightarrow \gamma(\mathbf{w}).$ Let $\hat{S}(s|\mathbf{w})$ and $\hat{\gamma}(\mathbf{w})$ denote the corresponding estimators.
\item Step 2: Let us recall that $\hat{\pi}_{\theta}(s)=(1+\tau)n^{-1}\sum_{i=1}^n X_{\theta,i}(s).$ Let $$\hat{\mathfrak{L}}^*_{\theta}(s)=\frac{1}{m}\sum_{j=1}^m \Psi\left(\hat{\pi}_{\theta}(s),\phi_{\theta}(\mathbf{W}_j);\hat{S}(s|\mathbf{W}_j),\hat{\gamma}(\mathbf{W}_j)\right),$$
where $\Psi$ is defined in (\ref{Epsi}). Then, define
\begin{equation}\label{eq:hattheta}\hat{\tilde{\theta}}^*(s)=\arg \max_{\theta \in \Theta}\hat{\mathfrak{L}}^*_{\theta}(s).
\end{equation}
\end{itemize}

The heuristic is the following: data on both $Y$ and $\mathbf{W}$ are rare (and potentially expensive, see \cite{andre2013contribution}). They are used to provide estimators of $\hat{S}(s|\cdot)$ and $\hat{\gamma}(\cdot).$ The convergence rate of these estimators will be limited by the sample size $n.$ On the other hand, we have plenty of information on the covariates $\mathbf{W},$ which means that $m$ is large. The errors on $\hat{S}$ and $\hat{\gamma}$ are averaged in (\ref{eq:hattheta}), which increases the performance. This improvement will be limited, first by the bias introduced by the fact that $\hat{\mathfrak{L}}^*_{\theta}(s)$ is an empirical approximation of $\mathfrak{L}^*_{\theta}(s)$ instead of $\mathfrak{L}_{\theta}(s)$. This bias is supposed to be controlled by the fact that $s$ should be selected high enough. A second limitation will become apparent with the conditions of section \ref{sec:crate}, but the rate of convergence of $\hat{\tilde{\theta}}^*(s)$ will be, at worse, the one of $\hat{S}(s|\mathbf{w})$ and $\hat{\gamma}(\mathbf{w}).$

Regarding possible estimators of $\hat{S}(s|\cdot)$ and $\hat{\gamma}(\cdot),$ examples will be provided in section \ref{sec:real}. A logistic regression can for example be used for $\hat{S}(s|\cdot).$ However, tail regression techniques may be used to simultaneously estimate this probability and the parameter $\hat{\gamma},$ see for example \cite{farkas2021cyber} and \cite{farkas2024generalized} who rely regression trees. The next section is then devoted to providing theoretical ground to legitimate the use of (\ref{eq:hattheta}).

\subsection{Convergence rate for the statistical optimization}
\label{sec:crate}

We here give theoretical elements to support the convergence of $\hat{\tilde{\theta}}^*(s)$ towards $\tilde{\theta}^*(s),$ where
$$\tilde{\theta}^*(s)=\arg \max_{\theta\in \Theta}\mathfrak{L}^*_{\theta}(s).$$ Let us note that $\tilde{\theta}^*(s)$ is different from $\tilde{\theta}(s)$ due to the approximation that is done in $\mathfrak{L}^*_{\theta}(s),$ but Theorem \ref{th_approx} shows that this bias can be contained if $s$ is large enough. This convergence requires a set of standard assumptions which are the adaptation of standard in (semi-)parametric M-estimation. Essentially, they require:
\begin{itemize}
\item regularity of the functions $\theta\rightarrow \mathfrak{L}^*_{\theta}(s),$ which is essentially a matter of regularity of the function $L$ and $f,$ see Assumption \ref{a_uf} below;
\item the estimation technique used for $\hat{S}$ and $\hat{\gamma}$ should be consistent and not too "unstable". This is the meaning of the more technical Assumptions \ref{a_donsker} and \ref{a_donskerbis}, see also the discussion below.
\end{itemize}

Let us also stress that we are focusing on a situation where $k$, the dimension of $\Theta,$ is small compared to $n$ and $m,$ to rely on asymptotic theory (for $n$ and $m$ tending to infinity with $n$ potentially much smaller than $m$). Hyper-parametrized techniques like deep neural networks are clearly beyond the scope of these theoretical results.

We now list the assumptions required for Theorem \ref{th_theta}.

\begin{assum}
\label{a_uf}
We assume the following conditions:
\begin{enumerate}
\item the functions $|L|,$ $|L'|,$ $|f'|$ and $\varphi_0$ and $\varphi_1$ are bounded;
\item there exists a constant $M_0$ such that 
$$\sup_{\theta\in \Theta} \left|\frac{\partial \pi_{\theta}(s)}{\partial \theta}\right|\leq M_0;$$
\item there exists a function $\Lambda$ with
$E[\Lambda(\mathbf{W})]<\infty$ such that
$$\left|\frac{\partial \phi_{\theta}(\mathbf{w})}{\partial \theta}\right|\leq \Lambda(\mathbf{w}).$$
\item we have $\gamma_-=\inf_{\mathbf{w}\in \mathcal{W}}\gamma(\mathbf{w})\geq \frac{1}{\alpha},$ for some $\alpha>0,$ with  $$\int_0^{\infty} v^{\alpha}\varphi_1(v)dv<\infty.$$
\end{enumerate}
\end{assum}

One can link this assumption with Theorem 3.2.5 in \cite{van1996weak}. The idea is essentially to ensure that, if we knew exactly $S$ and $\gamma,$ the empirical M-estimation procedure would be $m^{1/2}-$consistent.

\begin{assum}
\label{a_donsker}
Let $\mathcal{W}$ be such that $\mathbb{P}(\mathbf{W}\in \mathcal{W})=1.$
We have
$$\sup_{\mathbf{w}\in \mathcal{W}}\left|\frac{\hat{S}(s|\mathbf{w})-S(s|\mathbf{w})}{1-S(s|\mathbf{w})}\right|+\sup_{\mathbf{w}\in \mathcal{W}}\left| \hat{\gamma}(\mathbf{w})-\gamma(\mathbf{w})\right|=O_P(\varepsilon_n),$$
with $\varepsilon_n$ tending to zero.
\end{assum}

Assumption \ref{a_donsker} is related to the consistency of the preliminary method of estimation of the conditional tail distribution of $Y.$ Convergence rates of this type can be found in \cite{farkas2024generalized} for the case of estimation via regression trees. 

Assumption \ref{a_donskerbis} is more technical, and is related to the concept of Donsker class, see \cite{van1996weak}, Chapter 2.1 for a definition. 

\begin{assum}
\label{a_donskerbis}
\begin{itemize}
\item[1.]
The functions $\mathbf{w}\rightarrow \hat{S}(s|\mathbf{w})$ (resp. $\mathbf{w}\rightarrow \hat{\gamma}(\mathbf{w})$) and $\mathbf{w}\rightarrow S(s|\mathbf{w})$ (resp. $\mathbf{w}\rightarrow \gamma(\mathbf{w})$) belong to a Donsker class of functions $\mathcal{H}_1$ (resp. $\mathcal{H}_2$). Moreover there exists a constant $C$ such that all functions in $\mathcal{H}_1$ are bounded by $C$.
\item[2.] The classes of functions $\left\{\mathbf{w}\rightarrow \phi_{\theta}(\mathbf{w}):\theta\in \Theta\right\}$ and $\left\{\mathbf{w}\rightarrow \partial_{\theta}\phi_{\theta}(\mathbf{w})/\partial\theta:\theta\in \Theta\right\}$ are Donsker. Moreover the second class is bounded.
\end{itemize}
\end{assum}

The intuition behind the concept of Donsker classes, is that the estimators $\hat{S}(s|\mathbf{w})$ and $\hat{\gamma}(\mathbf{w})$ do not evolve in too complex classes of functions (point 1), that would be too large for a uniform central limit property to hold. A problem could occur if the estimators take too elaborate forms, or are discontinuous. Typically, it will hold for differentiable functions (with bounded derivatives up to some order depending on the dimension of $\mathbf{w},$ see for example Theorem 2.7.1 in \cite{van1996weak} ), or piecewise constant functions. For point 2, it will essentially hold if $\phi_{\theta}(\mathbf{w})$ and its gradient with respect to $\theta$ are Lipschitz with respect to $\theta$ (see Theorem 2.7.11 in \cite{van1996weak}), and if the dimension of $\Theta$ can be considered as negligible with respect to $m.$

We now state the main result of this section.

\begin{theorem}
\label{th_theta}
Let

$$B_n=\int \nabla_{\theta}\Psi\left(\pi_{\tilde{\theta}^*}(s),\phi_{\tilde{\theta}^*(s)}(\mathbf{w});\hat{S}(s|\mathbf{w}),\hat{\gamma}(\mathbf{w})\right)d\mathbb{P}_{\mathbf{W}}(\mathbf{w}),$$

where $\mathbb{P}_{\mathbf{W}}$ is the distribution of $\mathbf{W},$
and
$$\Sigma(s)=\nabla^2_{\theta}\mathfrak{L}^*_{\tilde{\theta}^*}(s),$$
which is assumed to be invertible.
Under the Assumptions of Theorem \ref{th_approx} and Assumptions \ref{a_uf} to \ref{a_donskerbis},
$$\hat{\tilde{\theta}}^*(s)-\tilde{\theta}^*(s)=-\Sigma(s)^{-1}\left\{\frac{1}{m}\sum_{j=1}^m \nabla_{\theta}\Psi\left(\pi_{\tilde{\theta}^*}(s),\phi_{\tilde{\theta}^*(s)}(\mathbf{W}_j);S(s|\mathbf{W}_j),\gamma(\mathbf{W}_j)\right)\right\}+O_P(B_n)+o_P(m^{-1/2}).$$
\end{theorem}

The proof is provided in section \ref{sec:proofestim}.

In this decomposition, the empirical sum is tending to 0 at rate $m^{-1/2}.$
The term $B_n$ is tending to zero, since, from the definition of $\tilde{\theta}^*(s),$
$$\int \nabla_{\theta}\Psi\left(\pi_{\tilde{\theta}^*}(s),\phi_{\tilde{\theta}^*(s)}(\mathbf{w});S(s|\mathbf{w}),\gamma(\mathbf{w})\right)d\mathbb{P}_{\mathbf{W}}(\mathbf{w})=0,$$
and $\hat{S}$ and $\hat{\gamma}$ converge towards $S$ and $\gamma.$ The rate convergence will be at most $\varepsilon_n,$ but if $\nabla_{\theta}\Psi\left(\pi_{\tilde{\theta}^*}(s),\phi_{\tilde{\theta}^*(s)}(\mathbf{w});\hat{S}(s|\mathbf{w}),\hat{\gamma}(\mathbf{w})\right)$ is an unbiased estimator of its limit, $B_n=0$ thanks to the integration with respect to $\mathbb{P}_{\mathbf{W}}(\mathbf{w})$ (and if the bias is small, it can lead to a significant gain). This property will be confirmed by the empirical example studied in the next section.

\section{Empirical analysis on simulated data}
\label{sec:simul}

In this section, we will illustrate the theoretical results obtained so far on simulated data. The work in this paper applies to heavy-tailed losses, which are the focus of all the analyses in this section. Section \ref{sec:simul_set} presents the simulation setting, the preliminary manipulations of the data, and the analysis algorithm used. Section \ref{sec:simul_res} presents the results obtained from these analyses. 

\subsection{Setting} \label{sec:simul_set}
These analyses have two main objectives. The first is to perform an empirical comparison of the objective function $\mathfrak{L}_{\theta}$ and its approximation $\mathfrak{L}^*_{\theta}$ on empirical data for a given threshold $s$. The second is to empirically verify that a two-step estimation of the objective function, using the approximation $\mathfrak{L}_{\theta}^*$ and an influx of index data, could help improve the quality of the estimation.

For these analyses, we choose a framework with a one-dimensional index $\mathbf{W}$ and a heavy-tailed loss $Y$. $\mathbf{W}$ is assumed to follow a uniform distribution on $[0, 1]$. $Y$ is assumed to depend on the index $\mathbf{W}$ and to follow a Pareto distribution whose survival function is given by:

\begin{equation*}
    S(Y=t|\mathbf{W}=\mathbf{w}) = \mathbb{P}(Y>t|\mathbf{W}=\mathbf{w}) = \frac{1}{t^{\frac{1}{\gamma(\mathbf{w})}}}.
\end{equation*}

The link between $Y$ and $\mathbf{W}$ is established through the shape parameter $\gamma(\mathbf{W})$, which is defined as:

\begin{equation}\label{eq:simul_shape}
    \gamma(\mathbf{w}) = e^{-a - b\mathbf{w}},
\end{equation}

where $a$ and $b$ are parameters chosen such that $\gamma(\mathbf{W}) \in [0.2, 0.7]$ for all $W \in [0, 1]$. 

The threshold $s$ is chosen as the 85$^{\text{th}}$ percentile of the losses $Y$. In other words, the proposed index reinsurance is applied to the top 15\% of losses. Moreover, the loading factor $\tau$ is 10\%. The family of payoff functions $\phi_{\theta}(\mathbf{W})$ is defined as:
\begin{equation*}
    \phi_{\theta}(\mathbf{w}) = \frac{1}{s}\max\{\min[\mathbf{E}(Y|Y>s,\mathbf{W}=\mathbf{w}), e^{\theta \mathbf{w}}],s\}.
\end{equation*}

In this expression, the minimum between $\mathbf{E}(Y|Y>s,\mathbf{W}=\mathbf{w})$ and $e^{\theta \mathbf{w}}$ is used to ensure there is no overcompensation, at least on average. In practice, this adjustment could be modified if it results in an unacceptable proportion of under-compensation. Such modifications could involve adding a loading parameter to $\mathbf{E}(Y|Y>s,\mathbf{W}=\mathbf{w})$ or adapting the variable part $e^{\theta \mathbf{w}}$ and the domain of $\theta$ to implicitly account for and prevent overcompensation. The maximum between the previous expression and $s$ is used to guarantee that there is no under-compensation above $s$. This ensures that a policyholder who suffers a loss greater than $s$ receives at least $s$ as compensation.

The optimal value of $\theta$ is therefore obtained under the constraint that, for losses above $s$, the policyholder or the cedent (in the case of index reinsurance) receives compensation of at least $s$ and at most $\mathbf{E}(Y|Y>s,\mathbf{W}=\mathbf{w})$ (which is finite since $\gamma(\mathbf{w})$ stays below 1).

We consider $L(x) = L_2(x)=-\exp(-\mu x),$ where $\mu$ is set to $1.5$. For the price aversion function, we consider $f= f_r$, where $f_r$ is as defined in section \ref{sec:choicef}, with $\kappa = 1.415$ and $\beta = 1.65$. Note that, in practice, these values would be calibrated using information from the portfolio and the policyholders.

\subsection{Results} \label{sec:simul_res}

Figure \ref{fig:util_payout} presents an analysis of the objective function, $\hat{\mathfrak{L}}$\footnote{$\hat{\mathfrak{L}}_\theta(s)=\frac{1}{m}\sum_{i=1}^{m}\left[L\left(\frac{Y_i\mathbf{1}_{T(\mathbf{W}_i) \leq s}+s\phi_\theta(\mathbf{W}_i)\mathbf{1}_{T(\mathbf{W}_i)>s}}{Y_i}-f(\hat{\pi}_{\theta}(s))\right)\right]$}, and its approximation $\hat{\mathfrak{L}}^*$\footnote{$\hat{\mathfrak{L}}^*_{\theta}(s)=\frac{1}{m}\sum_{j=1}^m \Psi\left(\hat{\pi}_{\theta}(s),\phi_{\theta}(\mathbf{W}_j);\hat{S}(s|\mathbf{W}_j),\hat{\gamma}(\mathbf{W}_j)\right)$}, on the full sample of size $m=5000$. Plots of $\hat{\mathfrak{L}}$ and $\hat{\mathfrak{L}}^*$ are shown in panel (a). These curves will serve as references in our subsequent analyses. We observe that the curves of $\hat{\mathfrak{L}}$ and $\hat{\mathfrak{L}}^*$ are nearly identical, which empirically confirms the validity of the using $\hat{\mathfrak{L}}^*$ instead of $\hat{\mathfrak{L}}$. The slight difference between the two curves can be attributed mainly to the error introduced by the approximation of the integral present in the expression of $\hat{\mathfrak{L}}^*$ (see Theorem \ref{th_approx}).

\begin{figure}[!h]
    \centering
    \begin{subfigure}[b]{0.5\textwidth}
        \centering
        \includegraphics[width=\linewidth]{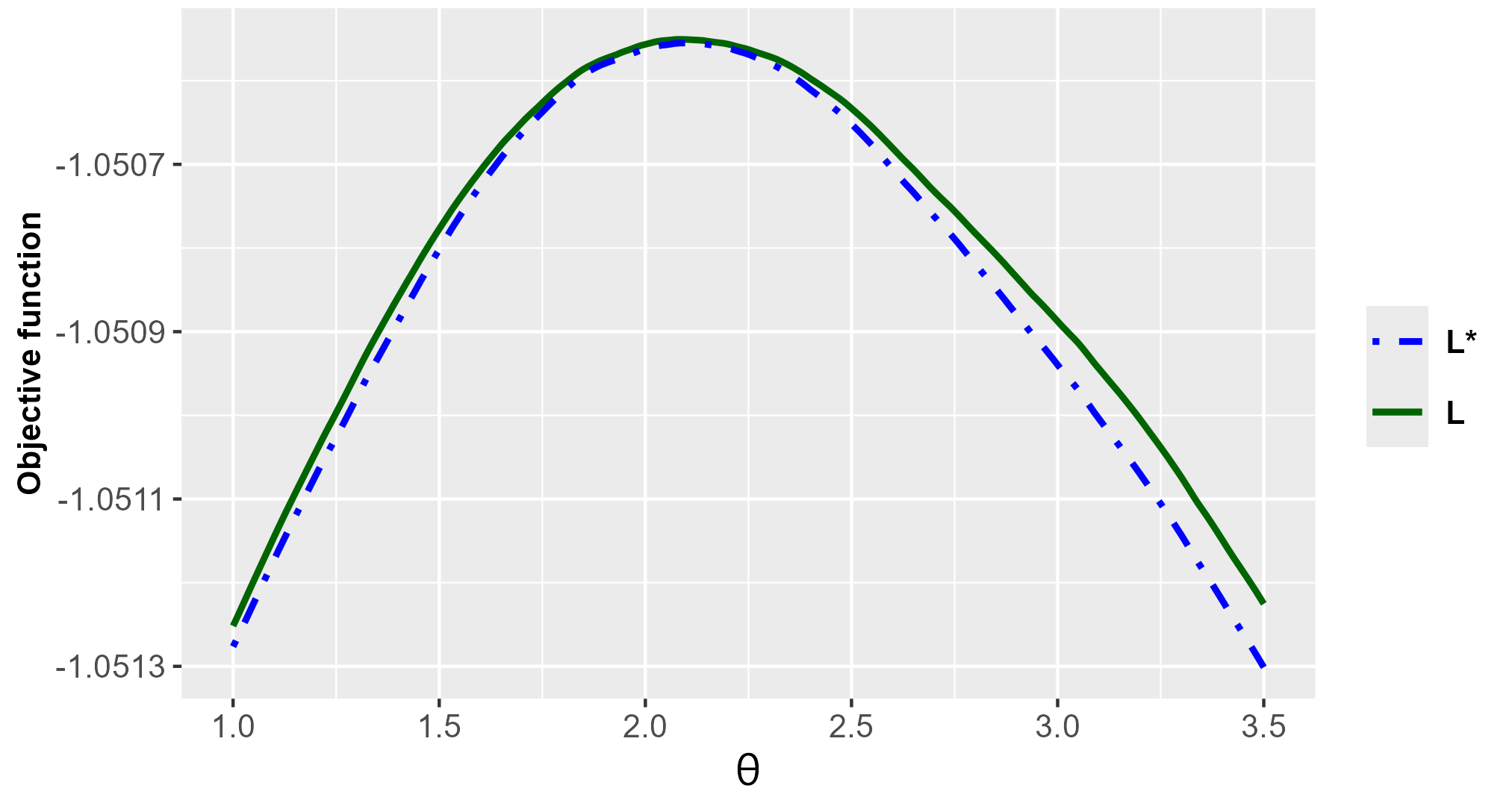}
        \caption{$\mathfrak{L}$ and $\mathfrak{L}^*$ on whole sample}
    \end{subfigure}%
    \begin{subfigure}[b]{0.5\textwidth}
        \centering
        \includegraphics[width=\linewidth]{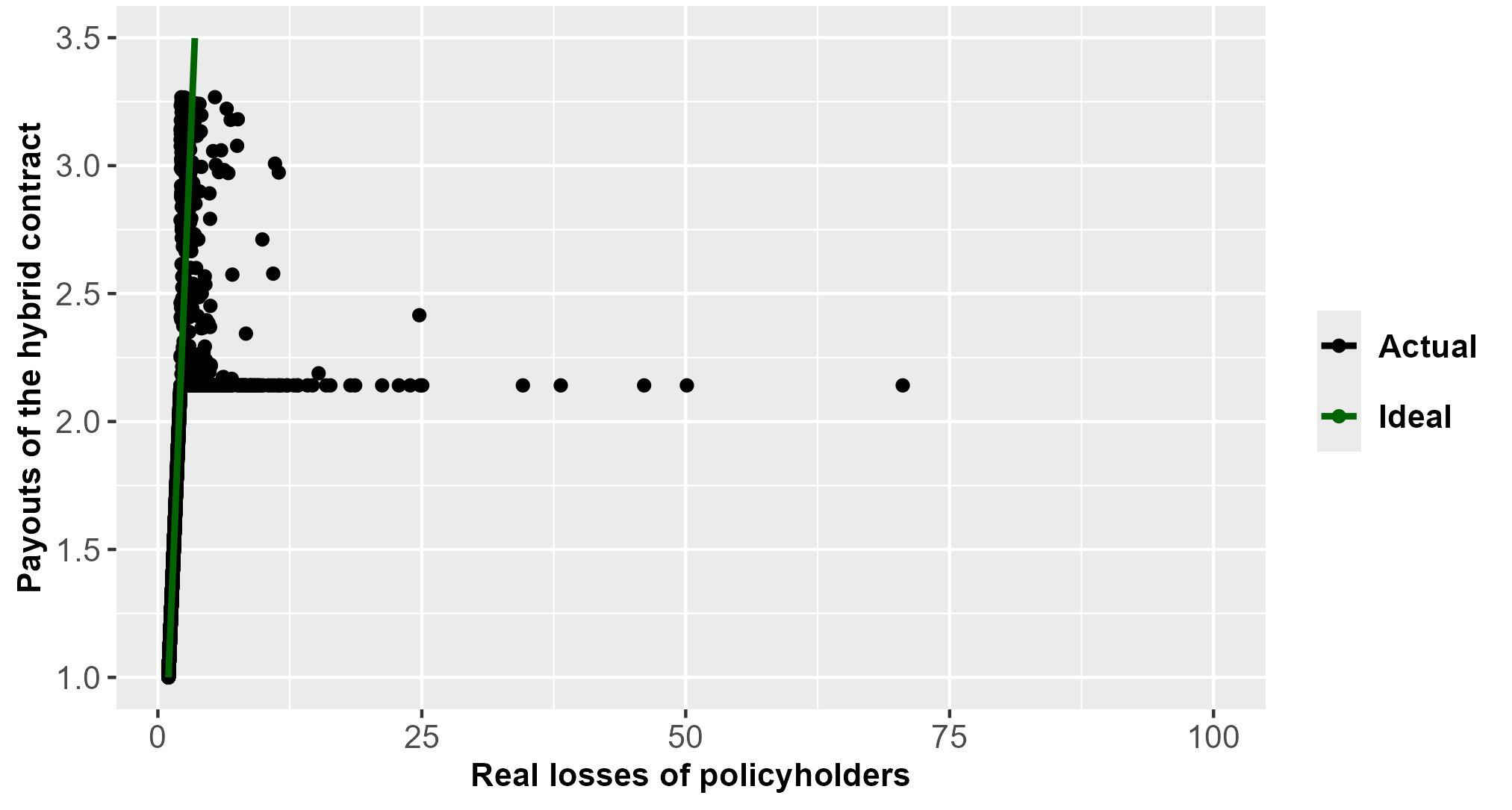}
        \caption{Payout for optimal $\theta$}
    \end{subfigure}
    \\[5pt]
    \caption{Panel (a) shows the objective functions for the whole sample of size $m$ as a function of the parameter $\theta$ for both $\hat{\mathfrak{L}}$ and $\hat{\mathfrak{L}}^*$. Panel (b) shows the payout function $\phi_{\hat{\tilde{\theta}}}$ as a function of the real losses of policyholders. $\hat{\tilde{\theta}}$ is the optimal value of $\theta$ obtained from the maximization of $\hat{\mathfrak{L}}_\theta$.}
    \label{fig:util_payout}
\end{figure}

Panel (b) shows the payoff function $\phi_{\theta}$ for the optimal value of $\theta$ as a function of the actual losses of the policyholders. The ideal payout function, which corresponds to the exact loss $Y$, is also plotted. We observe that the limits $s$ and $\mathbf{E}(Y|Y>s,\mathbf{W}=\mathbf{w})$ are quite effective, as they seem to successfully prevent undercompensation and overcompensation, respectively. The optimal value of $\theta$, which maximizes the empirical objective functions, produces a payout close to the ideal payout. This is evidenced by the clustering of the payouts above $s$ around the line $\phi_{\theta} = Y$.

The deviations from the ideal payout are not necessarily problematic, as they will be accounted for in the price of the insurance contract, ensuring that policyholders pay exactly for the amount of coverage they receive. Of course, this is guaranteed only if the index model is sufficiently accurate.

The next part of our analysis consists of studying and illustrating the advantages of a two-step optimization of $\theta$ using $\mathfrak{L}^*$. For this, according to the setting of section \ref{sec:alternative}, we consider a situation where we have a sample $(Y_i, \mathbf{W}_i)_{1\leq i \leq n}$ of size $n$ (where we will make this size $n$ vary) and a larger sample $(\mathbf{W}_j)_{1\leq j \leq m}$ of size $m > n$. Algorithm \ref{alg:simul} outlines the steps followed in this analysis.

\begin{algorithm}
    \caption{Methodology of analysis}\label{alg:simul}
    \begin{algorithmic}[1]
    \State $m \gets \text{total sample size}$
    \State $n \gets \text{a starting sample size}$
    \State $k \gets \text{an increment}$
    \State $\hat{\mathfrak{L}}_{ref}(\theta) \gets \frac{1}{m}\sum_{i=1}^{m}\left[L\left(\frac{Y_i\mathbf{1}_{T(\mathbf{W}_i) \leq s}+s\phi_\theta(W_i)\mathbf{1}_{T(\mathbf{W}_i)>s}}{Y_i}-f(\hat{\pi}_{\theta}(s))\right)\right]$
    \Repeat
        \State $\hat{\mathfrak{L}}_{n}(\theta) \gets \frac{1}{n}\sum_{i=1}^{n}\left[L\left(\frac{Y_i\mathbf{1}_{T(\mathbf{W}_i) \leq s}+s\phi_\theta(\mathbf{W}_i)\mathbf{1}_{T(\mathbf{W}_i)>s}}{Y_i}-f(\hat{\pi}_{\theta}(s))\right)\right]$
        \State $\text{Estimate } \hat{a}_n \text{ and } \hat{b}_n \text{ using } (Y_i,W_i)_{1 \leq i \leq n} \text{ to obtain } \hat{\gamma}(W) \text{ and } \hat{S}(s|W)$
        \State $\hat{\mathfrak{L}}_{n}^*(\theta) \gets L\left(1 - f(\hat{\pi}_{\theta}(s))\right)\frac{1}{m}\sum_{i=1}^{m}\left[1 - \hat{S}(s|W_i)\Phi_{0}\left(\phi_{\theta}(W_i), \hat{\gamma}(W_i)\right)\right]$
        \State $Error_n \gets \|\hat{\mathfrak{L}}_n, \hat{\mathfrak{L}}_{ref}\|_\infty = \sup_\theta|\hat{\mathfrak{L}}_n(\theta)-\hat{\mathfrak{L}}_{ref}(\theta)|$
        \State $Error_n^* \gets \|\hat{\mathfrak{L}}^*_n, \hat{\mathfrak{L}}_{ref}\|_\infty =  \sup_\theta|\hat{\mathfrak{L}}^*_n(\theta)-\hat{\mathfrak{L}}_{ref}(\theta)|$
        \State $\text{store } Error_n \text{, } Error_n^* \text{ and the parameters } \hat{a}_n \text{ and } \hat{b}_n$
        \State $n \gets n+k$
    \Until{$n=m$}
    \end{algorithmic}
\end{algorithm}

\begin{figure}[!h]
    \centering
    \begin{subfigure}[b]{0.5\textwidth}
        \centering
        \includegraphics[width=\linewidth]{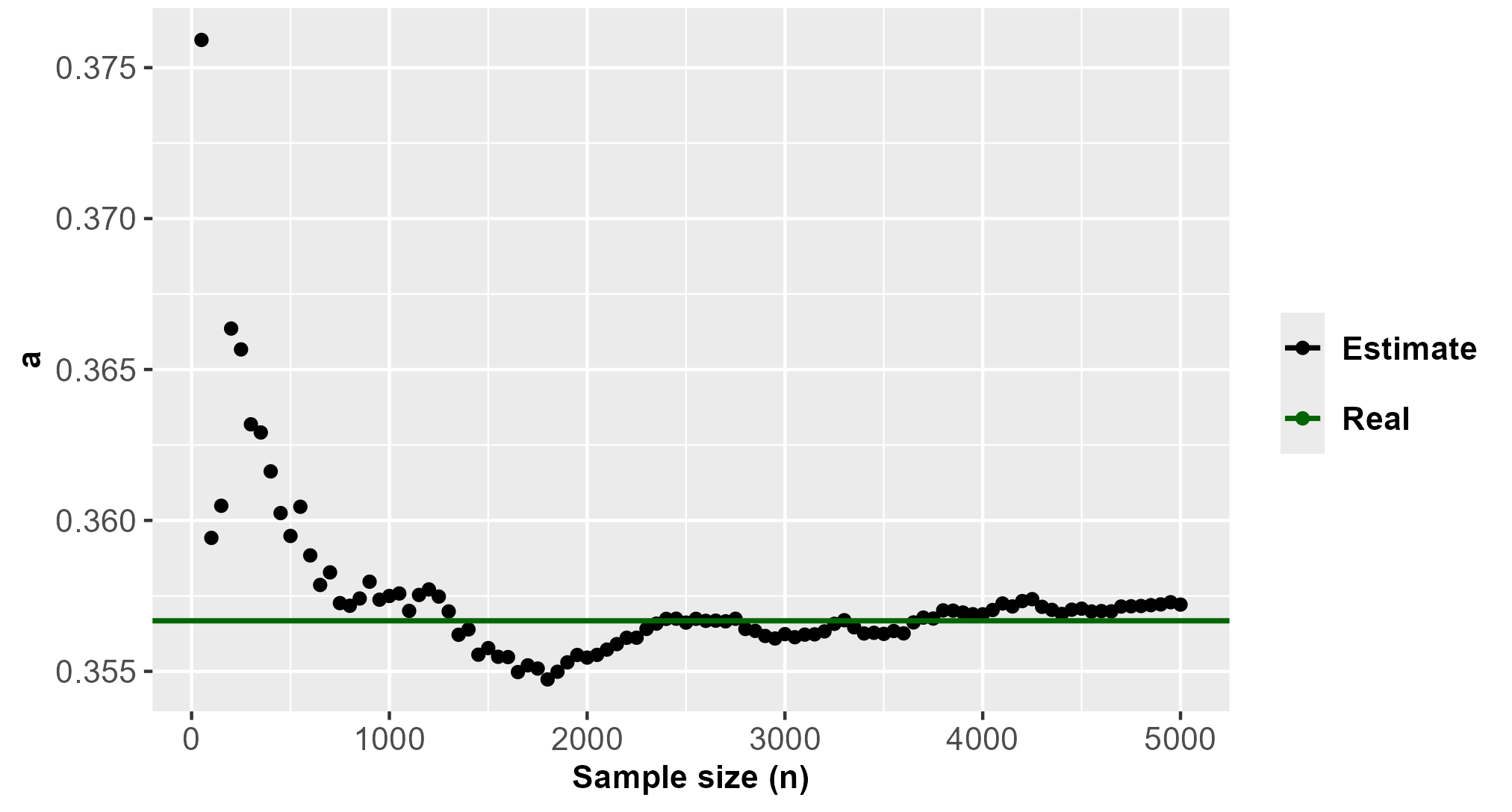}
        \caption{Estimates of the parameter $a$}
    \end{subfigure}%
    \begin{subfigure}[b]{0.5\textwidth}
        \centering
        \includegraphics[width=\linewidth]{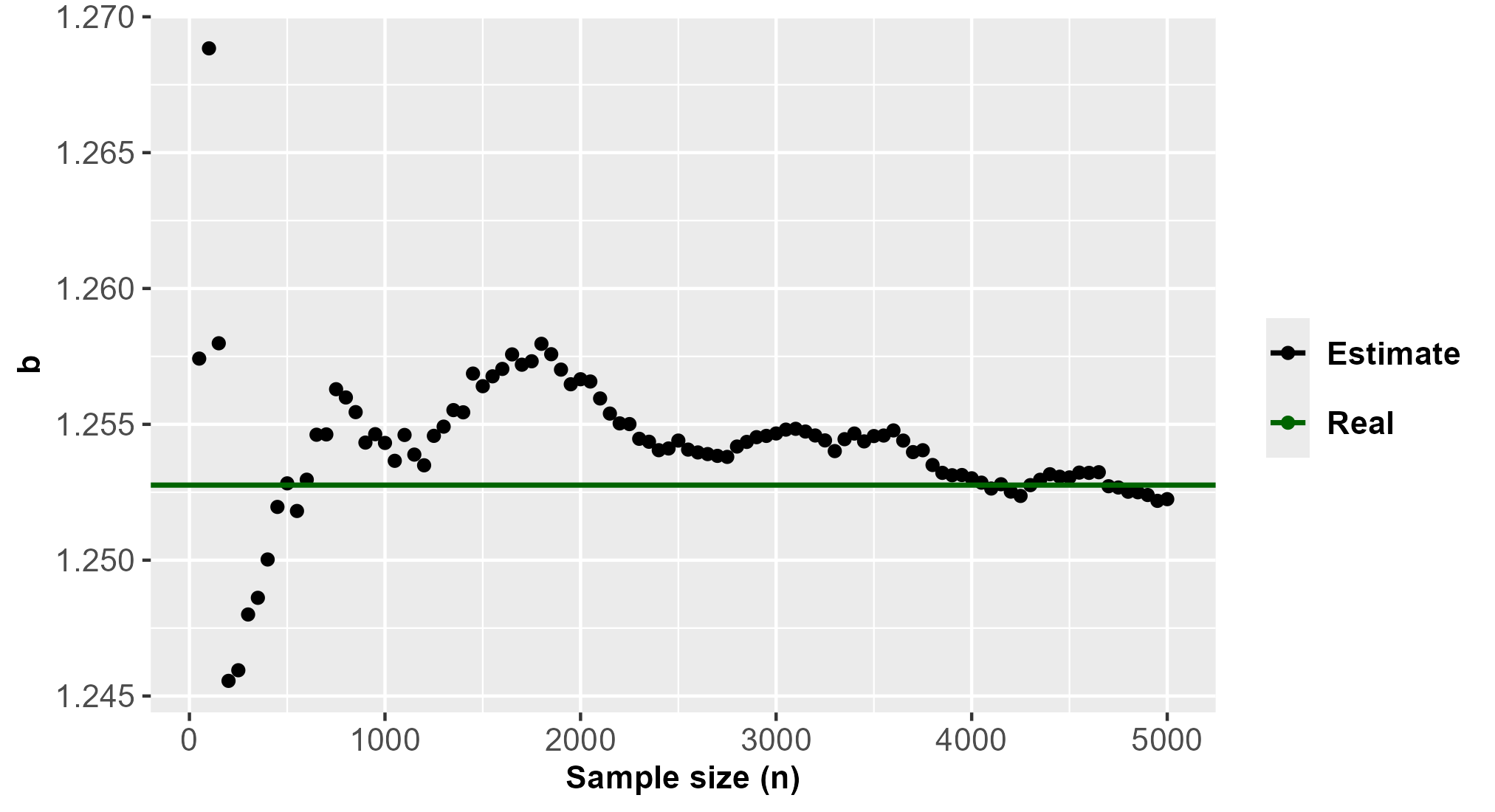}
        \caption{Estimates of the parameter $b$}
    \end{subfigure}
    \\[5pt]
    \caption{Evolution of the estimates of the parameters $a$ (panel (a)) and $b$ (panel (b)) of equation \ref{eq:simul_shape} with sample size.}
    \label{fig:param_simul}
\end{figure}

In step 7 of this algorithm, we use the smaller sample $(Y_i, \mathbf{W}_i)_{1\leq i \leq n}$ to estimate the parameters $a$ and $b$ of equation \ref{eq:simul_shape}. These estimated parameters are sufficient to define the functions $\mathbf{w} \mapsto \gamma(\mathbf{w})$ and $\mathbf{w} \mapsto S(s|\mathbf{w})$ from (\ref{eq:simul_shape}) and the definition of $S(t|\mathbf{w})$ from section \ref{sec:simul_set}. In step 8, these two functions are incorporated into the expression of $\hat{\mathfrak{L}}^*$, which does not depend on $Y$ or its entire distribution. $\hat{\mathfrak{L}}^*$ is an estimate of $\mathfrak{L}^*$ which is calculated using the larger sample $(\mathbf{W}_i)_{1\leq i \leq m}$.

Figure \ref{fig:param_simul} illustrates the evolution of the estimates of the parameters $a$ and $b$ as a function of the sample size $n$. As expected, as the size of the sample used for estimation increases, the estimates converge to the true values of these parameters. These true values, which serve as references, are estimated using the full sample of size $m$. 

\begin{figure}[!h]
    \centering
    \begin{subfigure}[b]{0.5\textwidth}
        \centering
        \includegraphics[width=\linewidth]{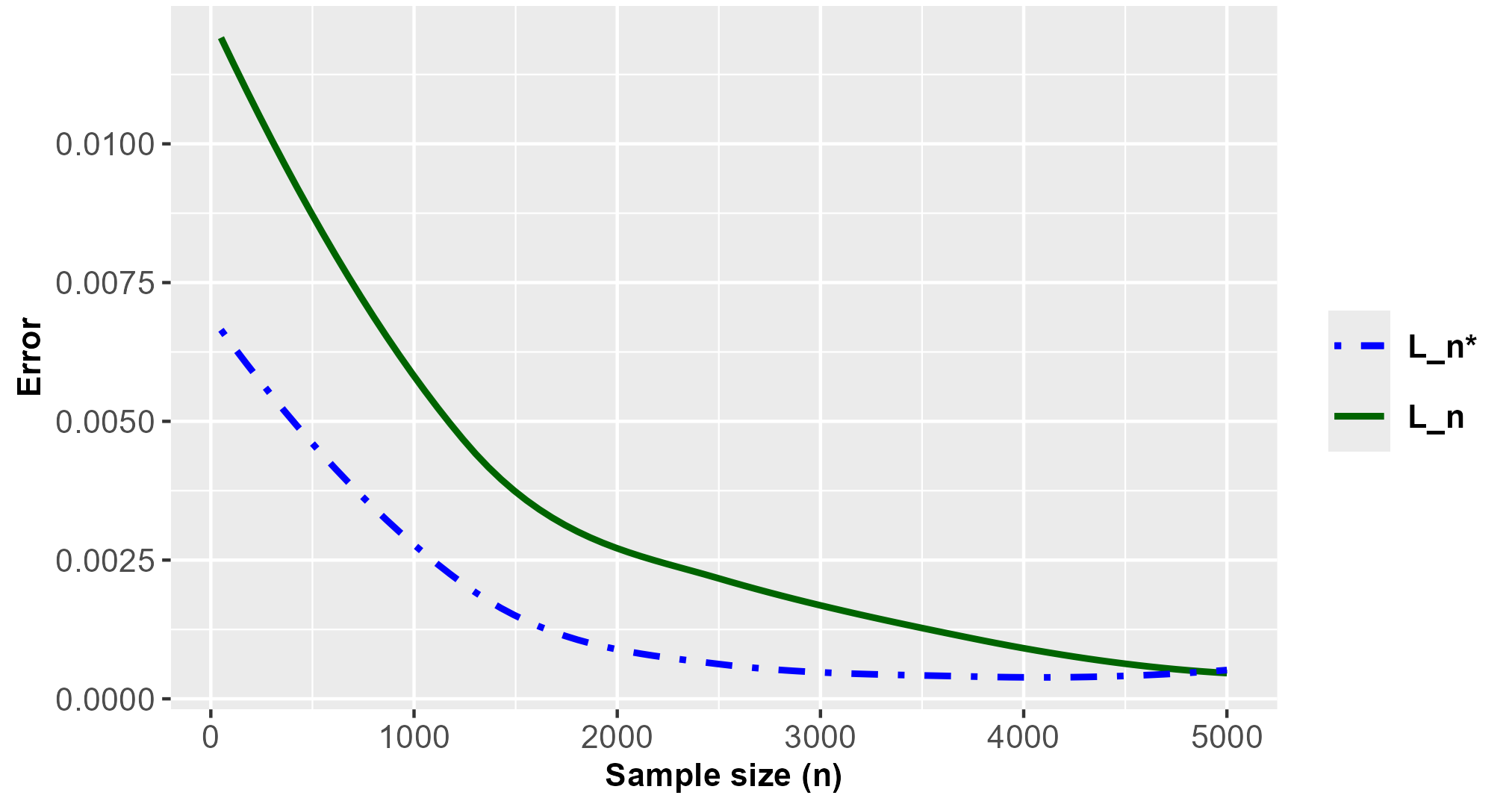}
        \caption{Errors in estimating $\hat{\mathfrak{L}}$ by $\hat{\mathfrak{L}}_n$ and $\hat{\mathfrak{L}}^*_n$.}
    \end{subfigure}%
    \begin{subfigure}[b]{0.5\textwidth}
        \centering
        \includegraphics[width=\linewidth]{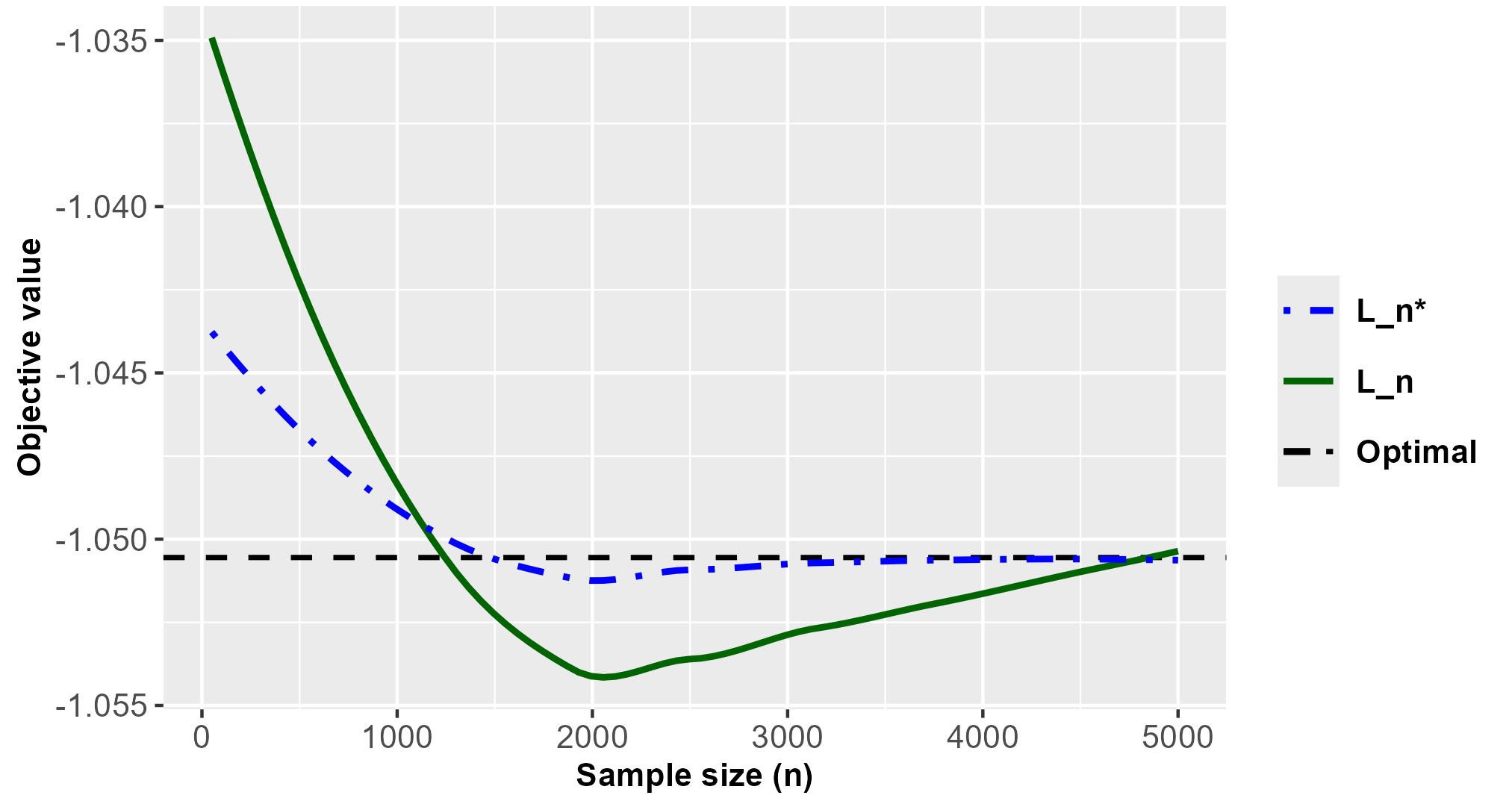}
        \caption{Optimal values of $\hat{\mathfrak{L}}_n$, $\hat{\mathfrak{L}}^*_n$ and $\hat{\mathfrak{L}}$.}
    \end{subfigure}
    \\[5pt]
    \caption{Panel (a) shows the errors incurred in the estimation of $\hat{\mathfrak{L}}$ by $\hat{\mathfrak{L}}_n$ (in one step) and $\hat{\mathfrak{L}}^*_n$ (in two steps) as a function of $n$ (see steps 7 and 8 of algorithm \ref{alg:simul}). Panel (b) shows the objective values obtained through $\hat{\mathfrak{L}} = \hat{\mathfrak{L}}_m$, $\hat{\mathfrak{L}}_n$, and $\hat{\mathfrak{L}}^*_n$. The evolution of the latter two is presented as a function of $n$.}
    \label{fig:error_simul}
\end{figure}

The results obtained by applying algorithm \ref{alg:simul} are presented in figure \ref{fig:error_simul}. In panel (a), we observe a decrease in the estimation errors of $\hat{\mathfrak{L}}$ by $\hat{\mathfrak{L}}_n$ (one-step estimation) and $\hat{\mathfrak{L}}^*_n$ (two-step estimation) as the sample size of $(Y_i, \mathbf{W}_i)_{1\leq i \leq n}$ increases. We also observe that the two-step estimation of the objective function using $\hat{\mathfrak{L}}^*_n$ results in a lower estimation error compared to the one-step estimation using $\hat{\mathfrak{L}}_n$. This result supports the conclusion that in situations with limited information on economic losses, the proposed methodology can improve the estimation of the optimal index insurance model parameters $\theta$ by leveraging a larger sample of index values.

In practice, beyond addressing the scarcity of economic loss data for emerging risks, this approach could also help in accounting for the future evolution of index values in current estimations, provided reliable projections of index values are available. Such an approach could significantly enhance the risk management strategies of index insurers and reinsurers.

The efficiency of the two-step approach using $\hat{\mathfrak{L}}^*$ is further corroborated by panel (b) of figure \ref{fig:error_simul}. Specifically, the objective values calculated in two steps using $\hat{\mathfrak{L}}^*_n$ are closer to, and converge faster toward, the reference objective value compared to the objective values calculated in one step using $\hat{\mathfrak{L}}_n$. 

In addition to reducing the number of random variables and the amount of information required to estimate the objective function $\hat{\mathfrak{L}}^*$, there may be an additional practical advantage to using $\hat{\mathfrak{L}}^*$. In a highly favorable scenario where the insurance or reinsurance company possesses a very large dataset on $Y$ and $\mathbf{W}$, estimating and optimizing $\hat{\mathfrak{L}}$ directly on this large dataset could be computationally expensive. In such a case, an alternative solution could involve estimating $\hat{\gamma}(\mathbf{W})$ and $\hat{S}(s|\mathbf{W})$ on a smaller subset of the dataset, incorporating these estimates into $\hat{\mathfrak{L}}^*$, and then performing the optimization using only $\mathbf{W}$.

\section{Empirical analysis on real data}
\label{sec:real}

In this section, we present the results of our analysis using real data. We consider tornado loss data from the United States covering the period 2016 to 2023. Section \ref{sec:real_set} describes this dataset and the setting used for our analysis, while section \ref{sec:rela_res} presents the results obtained.

\subsection{Data and setting}\label{sec:real_set}

The tornado losses dataset\footnote{Available at https://www.spc.noaa.gov/wcm//\#data} used in this study is identical to the one analyzed by \cite{daouia2023inference}. In this work, the authors confirmed the heavy-tailed nature of these losses over the period 2010 to 2016 by fitting a Generalized Pareto Distribution. We validate the heavy-tailed nature of the loss data by examining the plots presented in figure \ref{fig:ev_analysis}. The Hill plot shows an initial high variance, followed by a stable region and then a deviation. The mean excess plot is globally strictly increasing, and the QQ-plot indicates that the tail of the loss data is much heavier than that of the exponential distribution. Finally, we perform a Cramér–von Mises goodness-of-fit test, which fails to reject the null hypothesis of a Generalized Pareto Distribution, with a p-value of 0.902.

The variables of interest in this dataset include tornado monetary losses ($loss$), the starting and ending longitude and latitude of each tornado ($slon$, $elon$, $slat$, and $elat$), and the length and width of the area affected by the tornado ($len$ and $wid$). To address issues related to the scale of the losses, we work with losses normalized per unit area, defined as:

\begin{equation*}
    Y = \frac{loss}{len \times wid}.
\end{equation*}

Similarly to \cite{daouia2023inference}, we chose the tornado's average geographical position as indices, defined as:

\begin{equation*}
    \mathbf{W}^{(1)} = \frac{slat + elat}{2}, \text{ and } \mathbf{W}^{(2)} = \frac{slon + elon}{2},
\end{equation*}
and consider $\mathbf{W}=(\mathbf{W}^{(1)}, \mathbf{W}^{(2)}).$ Including additional indices and covariates could make sense in the practical design of an index insurance or reinsurance product for tornado losses. However, due to the lack of data and to ensure the simplicity of our analysis, we limit ourselves to this two dimensional vector of covariates.

Our final sample $(Y_i, \mathbf{W}_i)_{1\leq i\leq m}$ contains $m=4659$ observations and covers the entire USA. In particular, our dataset includes the costliest tornado, which occurred on October 20, 2019, causing EF3 damage in the Dallas suburbs of Texas and resulting in \$1.5 billion in losses. Note that, in this example, we do not have two samples, one with $(Y,\mathbf{W})$ and one with $\mathbf{W}$ only, but we will artificially remove some values of $Y_i$ to measure the performance of the approximation $\mathfrak{L}^*.$

Figure \ref{fig:data_real} presents the tornado losses per unit area over the study period. We distinguish between losses below the 85$^{\text{th}}$ percentile, which corresponds to our choice of $s$, and losses above this threshold. Despite a concentration of losses in the eastern regions, extreme tornadoes (losses above the 85$^{\text{th}}$ percentile) occur across nearly all states in the US. This observation leads to the broader conclusion that no geographical location is entirely immune to extreme losses.

\begin{figure}[!h]
    \centering
    \includegraphics[width=\linewidth]{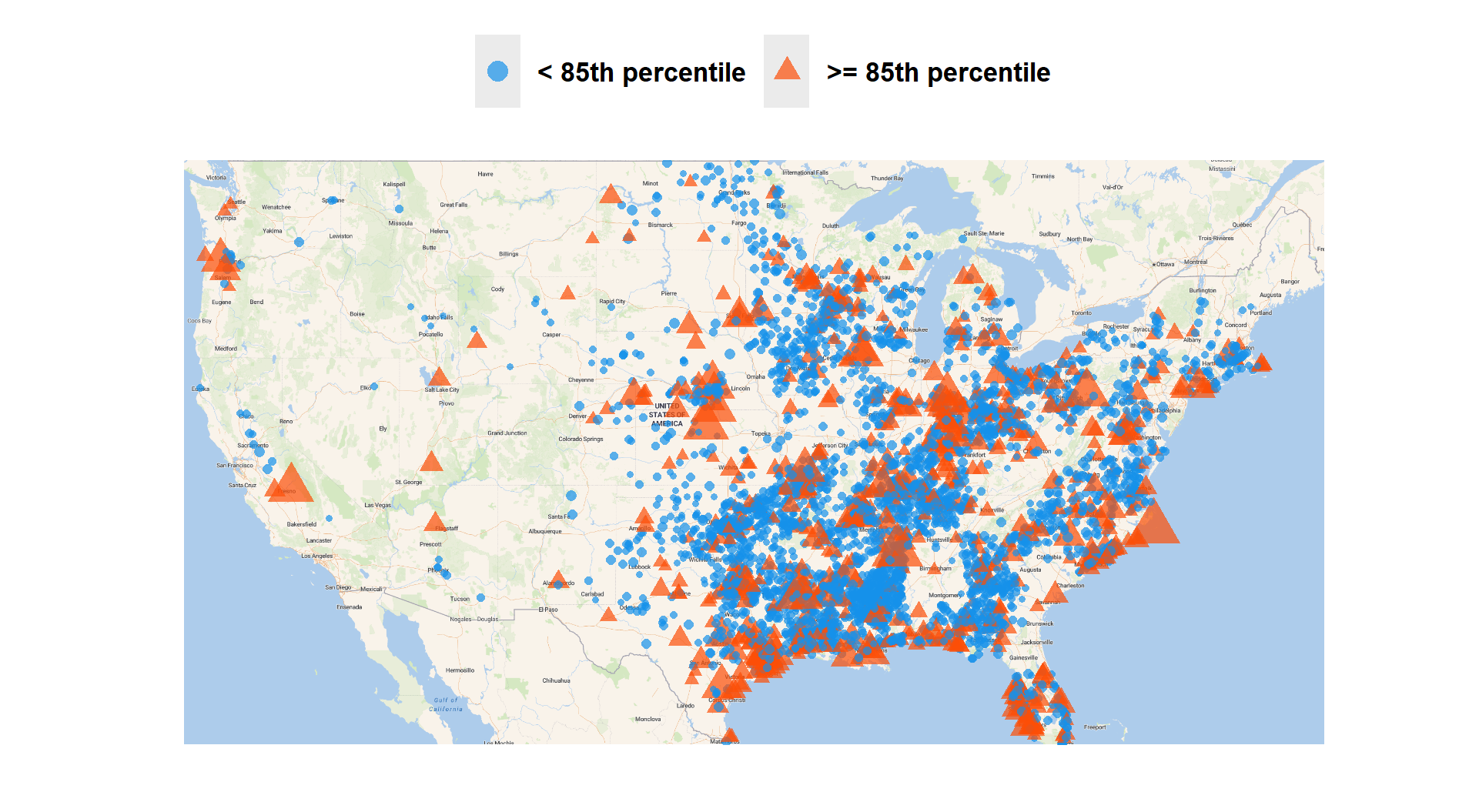}
    \vspace{-1.5cm}
    \caption{Tornado losses per unit area at the starting locations of the tornadoes. A distinction is made between losses per unit area below and above the 85$^{\text{th}}$ percentile, which is the value chosen for $s$.}
    \label{fig:data_real}
\end{figure}

Similarly to \cite{daouia2023inference}, and following the analysis of the data distribution described above, we assume a Generalized Pareto Distribution (GPD) for the tornado losses per unit area, with a dependence between these losses and the index values, such that:
\begin{equation*}
    S(Y=t|\mathbf{W}=\mathbf{w}) =\frac{1
    }{ \left(1+\frac{t \gamma(\mathbf{w})}{\sigma(\mathbf{w})}\right)^{\frac{1}{\gamma(\mathbf{w})}}}.
\end{equation*}

Where $\sigma$ is the scale parameter and $\gamma$ the shape parameter. $\gamma$ and $\sigma$ establish the statistical link between the losses $Y$ and the indices $\mathbf{W}^{(1)}$ and $\mathbf{W}^{(2)}$, through the equations:
\begin{equation}\label{eq:gamma_real}
    \begin{aligned}
        \gamma(\mathbf{w}) &= e^{-a-b\mathbf{w}^{(1)}-c\mathbf{w}^{(2)}} - 0.5,\\
        \sigma(\mathbf{w}) &= e^{-d-e\mathbf{w}^{(1)}-f\mathbf{w}^{(2)}}
    \end{aligned}
\end{equation}
where $a$, $b$, $c$, $d$, $e$ and $f$ are parameters to be estimated from the data. To validate this conditional GPD setting, we fit a GPD on the whole dataset with $\gamma$ and $\sigma$ defined by equation \ref{eq:gamma_real}. The results of this fit are presented in table \ref{tab:fitgpd} in the appendix. We perform a new Cramér--von Mises goodness-of-fit test, which fails to reject the null hypothesis of a conditional GPD, with a p-value of 0.5015. As shown in table \ref{tab:fitgpd}, all the fitted parameters ($\hat{a}$, $\hat{b}$, $\hat{c}$, $\hat{d}$, $\hat{e}$ and $\hat{f}$) are highly significant. Furthermore, a deviance comparative analysis between the fitted conditional GPD and a GPD with no covariates shows that the index $\mathbf{W}$ significantly contributes to explaining the variability of the tornado losses $Y$, including in the tails. This conclusion is also supported by the AIC values (more details can be found in table \ref{tab:deviancegpd} in the appendix).

As mentioned earlier, $s$ is chosen as the 85$^{\text{th}}$ percentile of the tornado losses per unit area. In practice, other choices of $s$ could be made depending on the setting and context without hindering the applicability of the proposed methodology.

The family of payoff functions is defined as:
\begin{equation*}
    \phi_{\theta_1, \theta_2}(\mathbf{w}) = \frac{1}{s}\max\{\min[q_{0.5}(Y|Y>s,\mathbf{W}=\mathbf{w}), \theta^{(1)} \mathbf{w}^{(1)} + \theta^{(2)} \mathbf{w}^{(2)}],s\}
\end{equation*}

where $q_{0.5}(Y | Y > s, \mathbf{W} = \mathbf{w})$ is the conditional median of losses per unit area above $s$. The expectation of $Y$ given $\mathbf{W}=\mathbf{w}$ according to the model is given by $\sigma(\mathbf{w})/(1-\gamma(\mathbf{w}))$. We do not directly use this form but prefer a linearized relationship between the payoff and $\mathbf{W}$ to avoid explosion of compensations for extremely heavy-tailed losses induced by the exponential relationship of equation \ref{eq:gamma_real}. As with the simulated data discussed in section \ref{sec:simul_set}, the conditional median is added to the payoff to reduce overcompensation. The choice of the conditional median, unlike the conditional expectation used for the simulated data, is motivated by the fact that, for certain tornado locations, the shape parameter may be greater than 1, rendering the conditional expectation undefined. Moreover, even restricting ourselves to areas where the expectation is finite, classical regression approaches (typically based on the empirical mean) are particularly unstable due to the heavy tail nature of the loss.

The choice of the median could be adapted in practice depending on the structure of the index model and the distribution of losses, as long as the selected statistic prevents overcompensation and limits under-compensation to an acceptable level. The justification of the general structure of the payout function is similar to that provided in section \ref{sec:simul_set} for simulated data.

The values of $\mu$, the parameters of the function $f$, and the loading factor of the premium are identical to those used for the simulated data. 

\subsection{Results}\label{sec:rela_res}

Panel (a) of Figure \ref{fig:res_eu_real} shows a plot of $\hat{\mathfrak{L}}^*$ with respect to $\theta_1$ and $\theta_2$ for $s = q_{0.85}(Y)$. We observe a strict concavity of $\hat{\mathfrak{L}}^*$ with respect to $\theta_1$ and $\theta_2$, suggesting the existence of an optimal combination $(\theta_1, \theta_2)$ that maximizes the objective function. Panel (b) of this figure shows plots of $\hat{\mathfrak{L}}^*$ for values of $s$ in ${q_{0.84}(Y), q_{0.85}(Y), q_{0.86}(Y)}$. We observe that as the value of the quantile increases, meaning as we focus on more extreme losses, the value of the objective function appears to increase.

This behavior is due to the implicit assumption that all losses below $s$ are fully and efficiently covered by indemnity-based insurance at the same cost as index insurance. In practice, this is not always the case, as indemnity-based insurance is not perfectly accurate. Furthermore, if we account for delays in compensation and high costs associated with indemnity-based insurance for extreme losses, it is reasonable to expect that, in practice, indemnity insurance would not always be preferred by policyholders for high or extreme losses.

\begin{figure}[!h]
    \centering
    \begin{subfigure}[b]{0.5\textwidth}
        \centering
        \includegraphics[width=\linewidth]{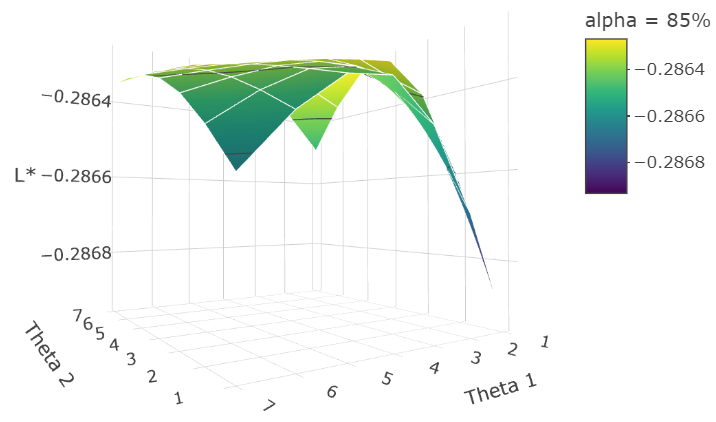}
        \caption{$\hat{\mathfrak{L}}^*$ for $s=q_{0.85}(Y)$}
    \end{subfigure}%
    \begin{subfigure}[b]{0.5\textwidth}
        \centering
        \includegraphics[width=\linewidth]{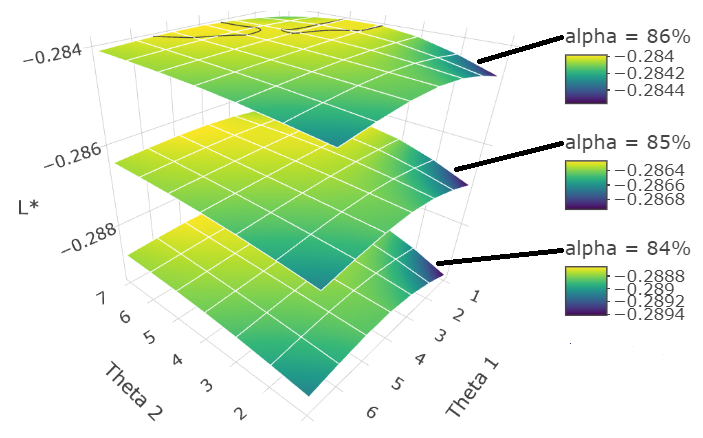}
        \caption{$\hat{\mathfrak{L}}^*$ for $s\in \{q_{0.84}(Y), q_{0.85}(Y), q_{0.86}(Y)\}$}
    \end{subfigure}
    \\[5pt]
    \caption{Objective function $\hat{\mathfrak{L}}^*$ with respect to $\theta_1$ and $\theta_2$ for $s=q_{0.85}(Y)$ (panel (a)) and for $s\in \{q_{0.84}(Y), q_{0.85}(Y), q_{0.86}(Y)\}$ (panel (b)).}
    \label{fig:res_eu_real}
\end{figure}

We pursue the analysis on real data by applying an algorithm similar to algorithm \ref{alg:simul}. Step 7 of this algorithm is slightly modified, as the three parameters $a$, $b$, and $c$ of equation \ref{eq:gamma_real} are estimated in this case. The evolutions of these estimates with respect to sample size are shown in Figure \ref{fig:param_real} in the Appendix.

We observe a significant divergence in the parameter estimates for small sample sizes, underscoring the critical importance of data availability in constructing efficient index insurance models. However, the estimates stabilize and become acceptable at approximately half of the total sample size. This finding supports our proposed two-step approach to optimizing the objective function, even in scenarios with limited data.

For instance, in a real life setting similar to this example, estimating the distribution parameters using half of the available data in the first step would yield acceptable estimates of $\hat{S}(s|W)$ and $\hat{\gamma}(W)$. These estimates would then enable an efficient calculation of the objective function in the second step, facilitating the development of an index insurance model in a situation where half of the needed loss data $Y$ is unavailable. 

\begin{figure}[!h]
    \centering
    \begin{subfigure}[b]{0.5\textwidth}
        \centering
        \includegraphics[width=\linewidth]{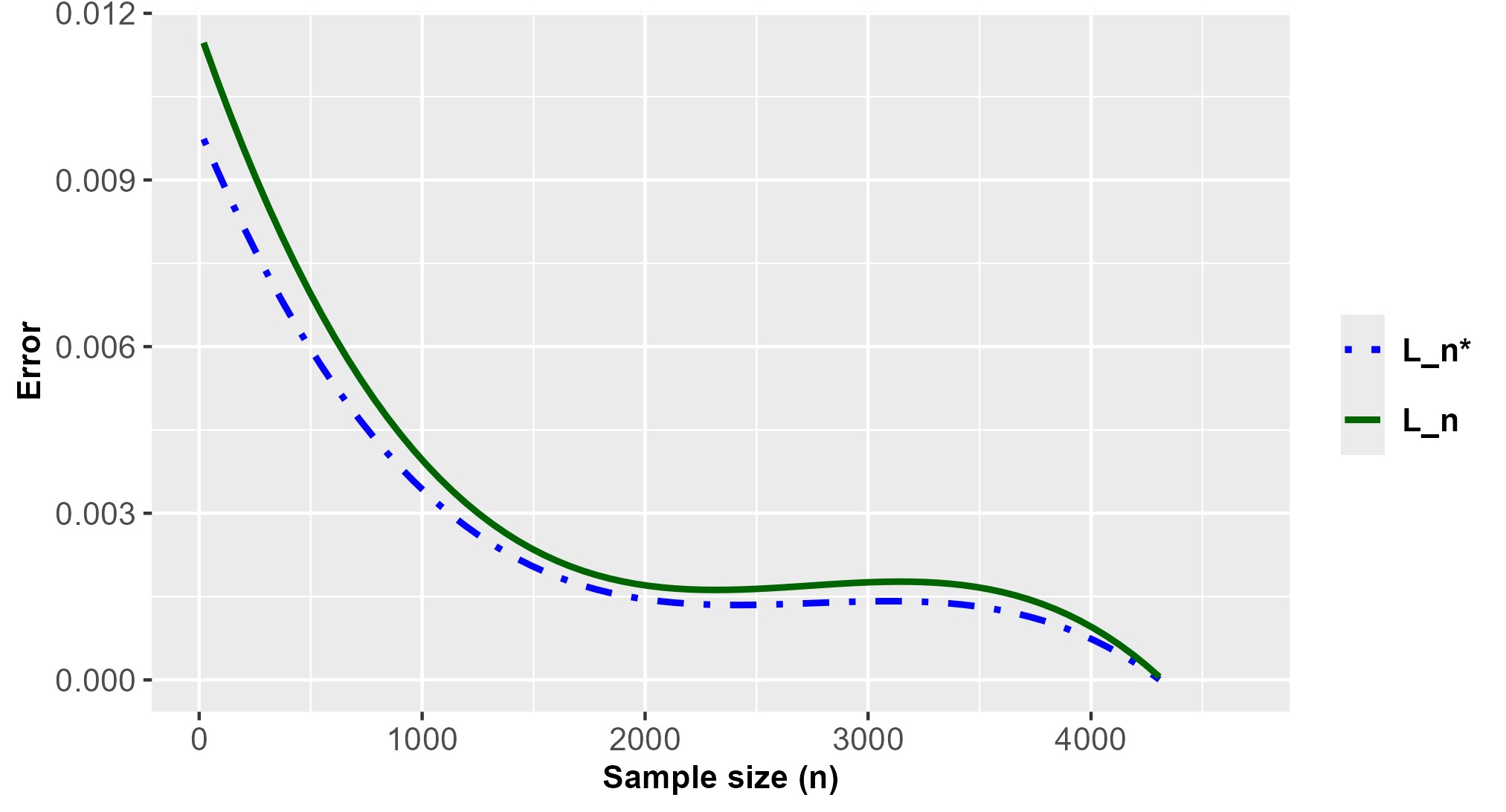}
        \caption{Using $X_{\theta}(s)$ and $\mathbf{1}_{T(\mathbf{W})>s}$ with $\mathfrak{E}_p = 0\%$}
    \end{subfigure}%
    \begin{subfigure}[b]{0.5\textwidth}
        \centering
        \includegraphics[width=\linewidth]{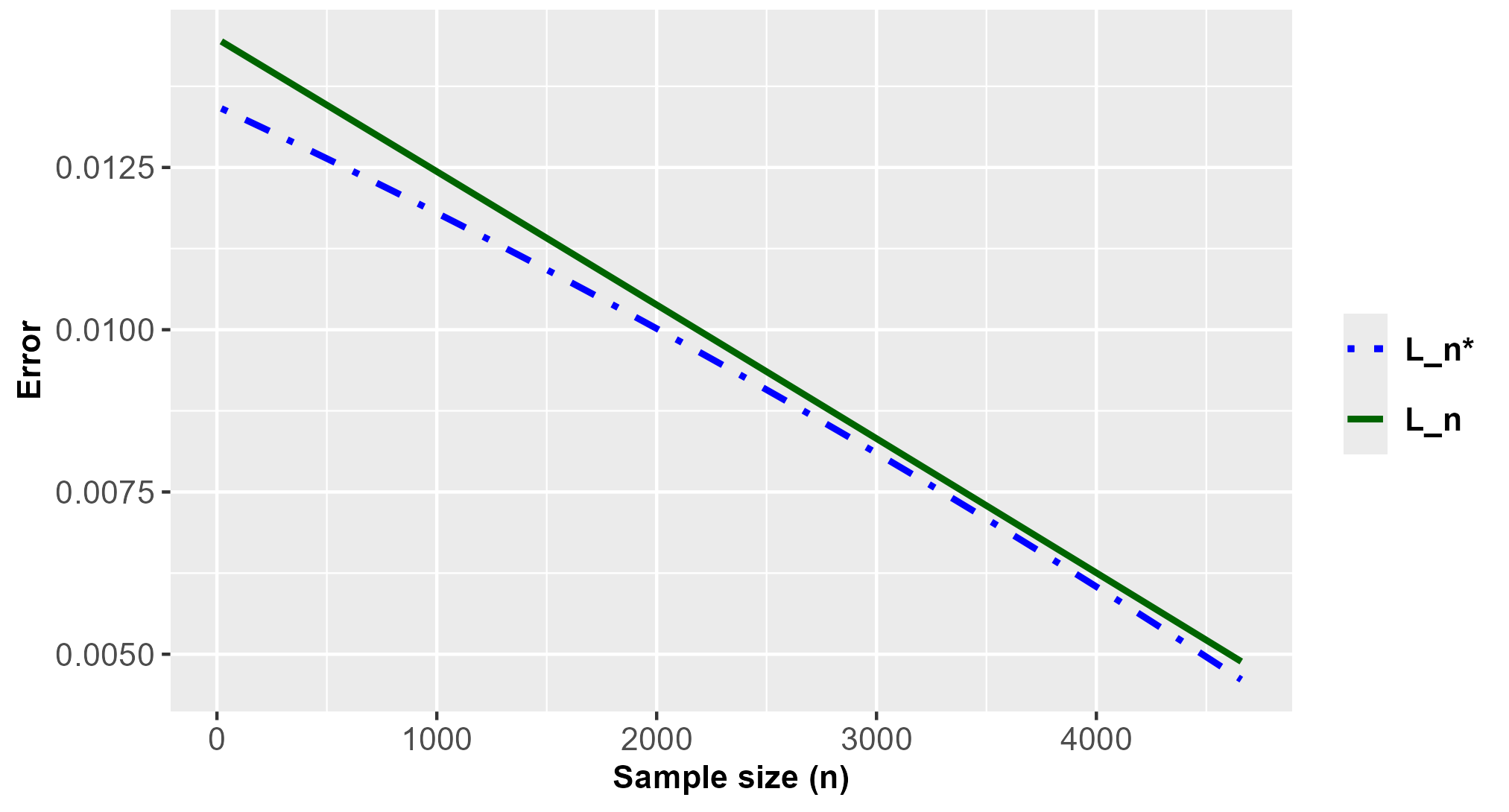}
        \caption{Using $X_{\theta}(s)$ and $\mathbf{1}_{T(\mathbf{W})>s}$ with $\mathfrak{E}_p = 5\%$}
    \end{subfigure}
    \\[5pt]
    \begin{subfigure}[b]{0.5\textwidth}
        \centering
        \includegraphics[width=\linewidth]{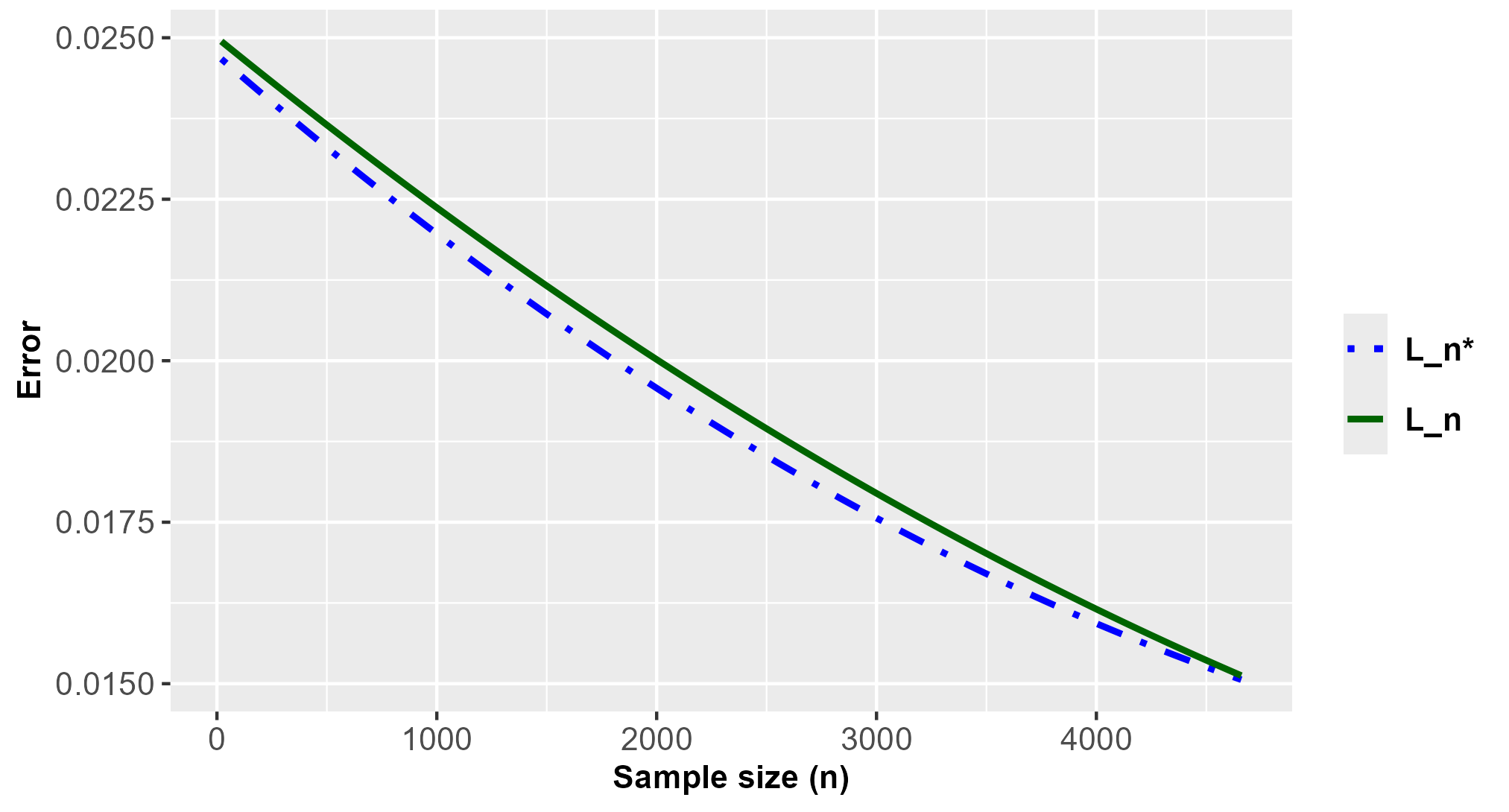}
        \caption{Using $X_{\theta}(s)$ and $\mathbf{1}_{T(\mathbf{W})>s}$ with $\mathfrak{E}_p = 10\%$}
    \end{subfigure}%
    \begin{subfigure}[b]{0.5\textwidth}
        \centering
        \includegraphics[width=\linewidth]{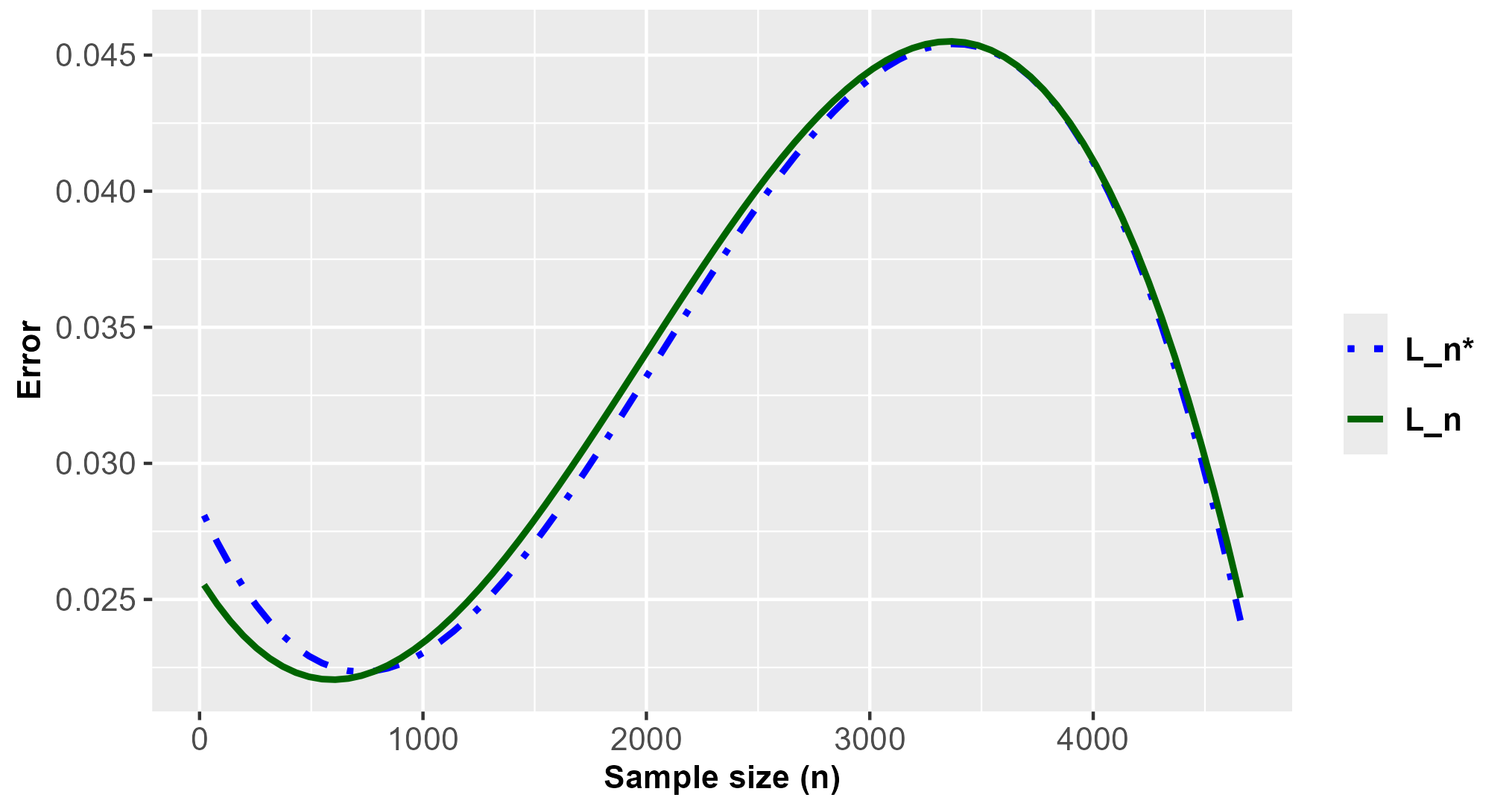}
        \caption{Using $X_{\theta}(s)$ and $\mathbf{1}_{T(\mathbf{W})>s}$ with $\mathfrak{E}_p = 15\%$}
    \end{subfigure}
    \\[5pt]
    \caption{Errors in estimating $\hat{\mathfrak{L}}$ by $\hat{\mathfrak{L}}_n$ and $\hat{\mathfrak{L}}_n^*$ as a function of sample size.}
    \label{fig:res_error_real}
\end{figure}

The results of applying algorithm \ref{alg:simul} to real data are presented in figure \ref{fig:res_error_real}. In this section, we compute the approximation error of $\hat{\mathfrak{L}}$ by $\hat{\mathfrak{L}}_n$ and $\hat{\mathfrak{L}}_n^*$ using $X_{\theta}(s)=Y\mathbf{1}_{T(\mathbf{W})\leq s}+s\phi_{\theta}(\mathbf{W})\mathbf{1}_{T(\mathbf{W})>s}$. In practice, $\mathbf{1}_{T(\mathbf{W})>s}$ could be estimated using a classification model. For this application, we reproduce the behavior of such a model by adding random noise to $\mathbf{1}_{Y>s}$ such that $\mathbf{1}_{T(\mathbf{W})>s} = \eta(\mathbf{1}_{Y>s}, \mathfrak{E}_p)$. Here, $\eta$ is the realistic noise function and $\mathfrak{E}_p$ is the desired prediction error, which we can control in order to study the effect of the classification model accuracy on our proposed methodology.

As expected, the panels of Figure \ref{fig:res_error_real} show that the errors incurred in estimating $\hat{\mathfrak{L}}$ by $\hat{\mathfrak{L}}_n$ and $\hat{\mathfrak{L}}_n^*$ increase as the accuracy of $\mathbf{1}_{T(\mathbf{W})>s}$ decreases. However, all panels of Figure \ref{fig:res_error_real} also show that the two-step estimation of $\hat{\mathfrak{L}}$ by $\hat{\mathfrak{L}}^*$ helps achieve lower errors and better estimates of the optimal expected utility, as observed with the simulated data. These results indicate that our proposed two-step methodology improves the estimation of the optimal index insurance contract in a real setting.

Note that the ability of our proposed contract to provide efficient coverage against protection gaps is possible only if a sufficient amount of data is available. Indeed, we observe that the lowest errors in estimating $\hat{\mathfrak{L}}$ are achieved only when the full sample size $m$ is considered.

In summary, the proposed index insurance model for extreme losses offers a practical solution for covering such losses and introduces an extension that addresses the problem of economic loss data scarcity—a major concern in insurance design. Moreover, in scenarios with scarce economic loss data $Y$, one could reasonably opt to work with $\hat{\mathfrak{L}^*}_{n}$ even if the relative approximation error is higher than that of $\hat{\mathfrak{L}}_{n}$, provided this error remains acceptable.

Data scarcity is particularly pronounced in developing countries, where data collection techniques and tools are often underdeveloped (\cite{eze2020exploring}). This lack of reliable data reduces the accuracy of insurance models in these regions, where certain insurance products could be vital given the significant risks faced. The methodology proposed in this paper addresses these challenges in two key ways:

\begin{itemize}
    \item The lower cost of index insurance used in our model makes it more accessible and suitable for use in developing countries.
    \item The two-step estimation of the objective function using $\hat{\mathfrak{L}}^*$ provides a practical approach to mitigate the issue of data scarcity prevalent in these regions.
\end{itemize}

\section{Comparison with a capped indemnity contract} \label{sec:comp_analysis}

In this section, we empirically investigate how the proposed hybrid contract positions itself compared to a capped indemnity contract. The capped indemnity contract (\cite{zhou2010optimal}, \cite{mao2021optimal}) is defined by the following payout:
\begin{equation*}
    X^{ST}(m) = Y\mathbf{1}_{Y \leq m} + m\mathbf{1}_{Y > m} = \min(Y, m).
\end{equation*}
Introducing a limit to compensation is a classical way to protect the insurer if the tail of the distribution is too heavy. On the other hand, this means that the policyholder will not benefit from full coverage against the risk, and this also has consequences on the attractiveness of the insurance product.

On the other hand, the hybrid product with payout defined as
\begin{equation*}
    X^{HB}(s)=Y\mathbf{1}_{T(\mathbf{W}) \leq s}+s\phi_{\theta}(\mathbf{W})\mathbf{1}_{T(\mathbf{W})>s},
\end{equation*}
is less limited in terms of maximum amount of compensation. The question is: for the same price, which type of contracts ($X^{ST}$ or $X^{HB}$) will give the best coverage for the policyholder, in terms of average ratio between the compensation and the loss? For the same price, the threshold $s$ of $X^{HB}$ is necessarily smaller than the value of $m$ used in $X^{ST},$ but this can be compensated by a better behavior of the parametric part. 

In the rest of this section, we use the optimal values of the parameter $\theta$ obtained from the computations in the previous sections. The settings used in those sections for simulated and real losses are equally preserved here. To generalize this analysis, we assume that in the hybrid contract, the loading factors of the index part and that of the indemnity part of the contract are different. This is expected in practice due to the lower operational costs of index insurance (\cite{barnett2007weather}). We denote by $\tau^i$ and $\tau^t$ the loading factors of index and traditional indemnity-based insurance, respectively.

The premium of the capped indemnity contract is therefore given by:

\begin{equation*}
    \pi^{SL}(m) = (1 + \tau^t)\mathbf{E}[X^{SL}] = (1 + \tau^t)\mathbf{E}[min(Y, m)]
\end{equation*}

That of the hybrid contract is given by:

\begin{equation*}
    \pi^{HB}(s) = (1 + \tau^t)\mathbf{E}[Y\mathbf{1}_{T(\mathbf{W}) \leq s}] + (1 + \tau^i)\mathbf{E}[s\phi_{\theta}(\mathbf{W})\mathbf{1}_{T(\mathbf{W})>s}]
\end{equation*}

We then define a function $m(s)$ such that, for a given $s$, $\pi^{HB}(s) = \pi^{SL}(m(s))$. The aim of this definition is to compare the expected rate of compensation of the capped indemnity contract $(\mathbf{E}[X^{SL}(m(s))/Y])$ to that of the hybrid contract $(\mathbf{E}[X^{HB}(s)/Y])$ at equal prices for a given $s$. For every $s$, the procedure consists in setting a value for $\tau^i$, calculating the value of $\pi^{HB}(s)$, and using the function defined previously to find $m(s)$ such that $\pi^{HB}(s) = \pi^{SL}(m(s))$. Note that as $\tau^i$ varies, the rate of compensation of the hybrid contract remains unchanged (as $s$ is fixed), but that of the capped indemnity contract has to adapt through $m(s)$ for the prices of the two contracts to be the same. The value of $\tau^t$ is kept fixed at 40\%. Figure \ref{fig:comp_analysis} shows the results obtained for simulated and real data. 

\begin{figure}[!h]
    \centering
    \begin{subfigure}[b]{0.5\textwidth}
        \centering
        \includegraphics[width=\linewidth]{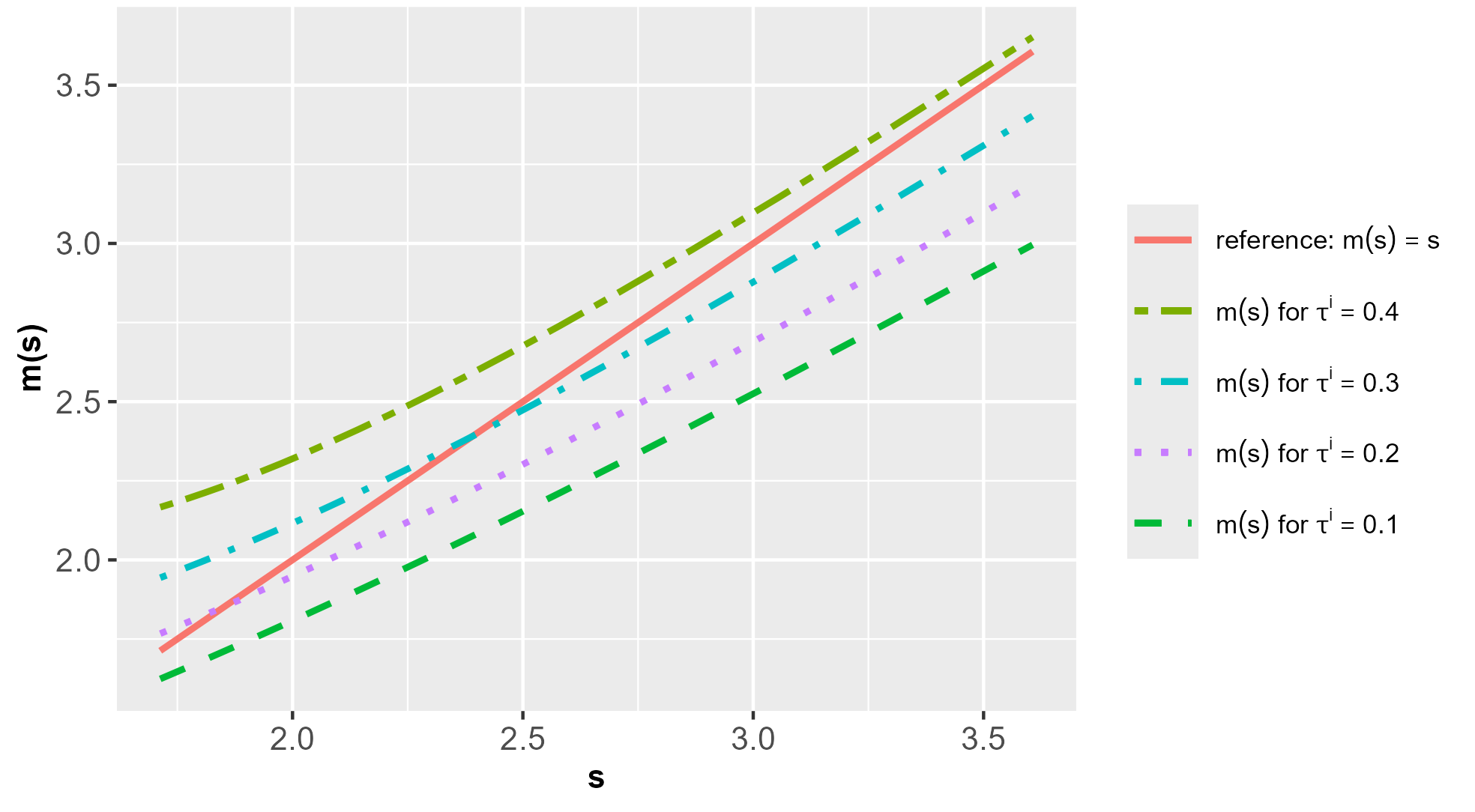}
        \caption{$m(s)$ vs $s$ for simulated losses}
    \end{subfigure}%
    \begin{subfigure}[b]{0.5\textwidth}
        \centering
        \includegraphics[width=\linewidth]{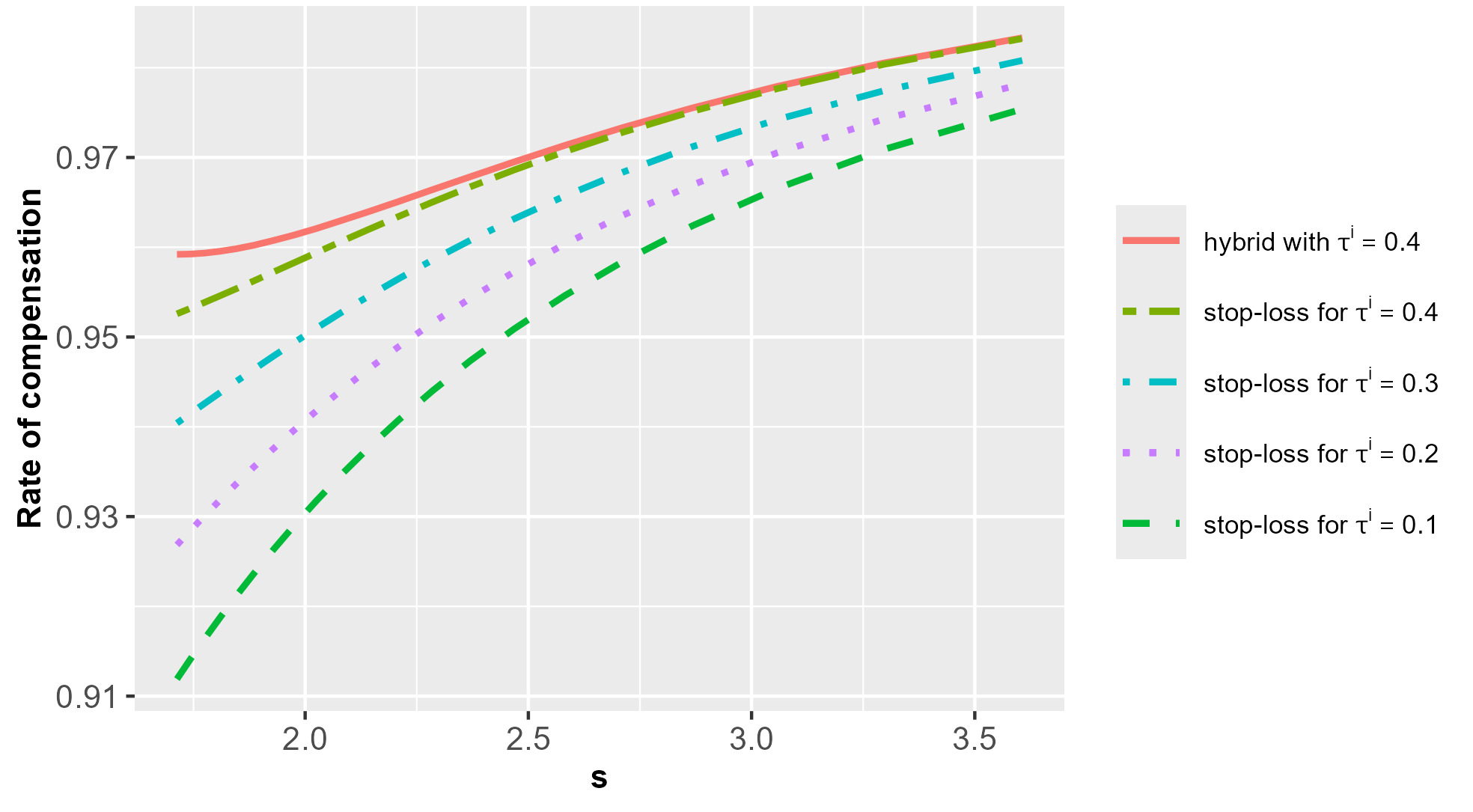}
        \caption{$\mathbf{E}[\frac{X}{Y}]$ vs $s$ for simulated losses}
    \end{subfigure}
    \\[5pt]
    \begin{subfigure}[b]{0.5\textwidth}
        \centering
        \includegraphics[width=\linewidth]{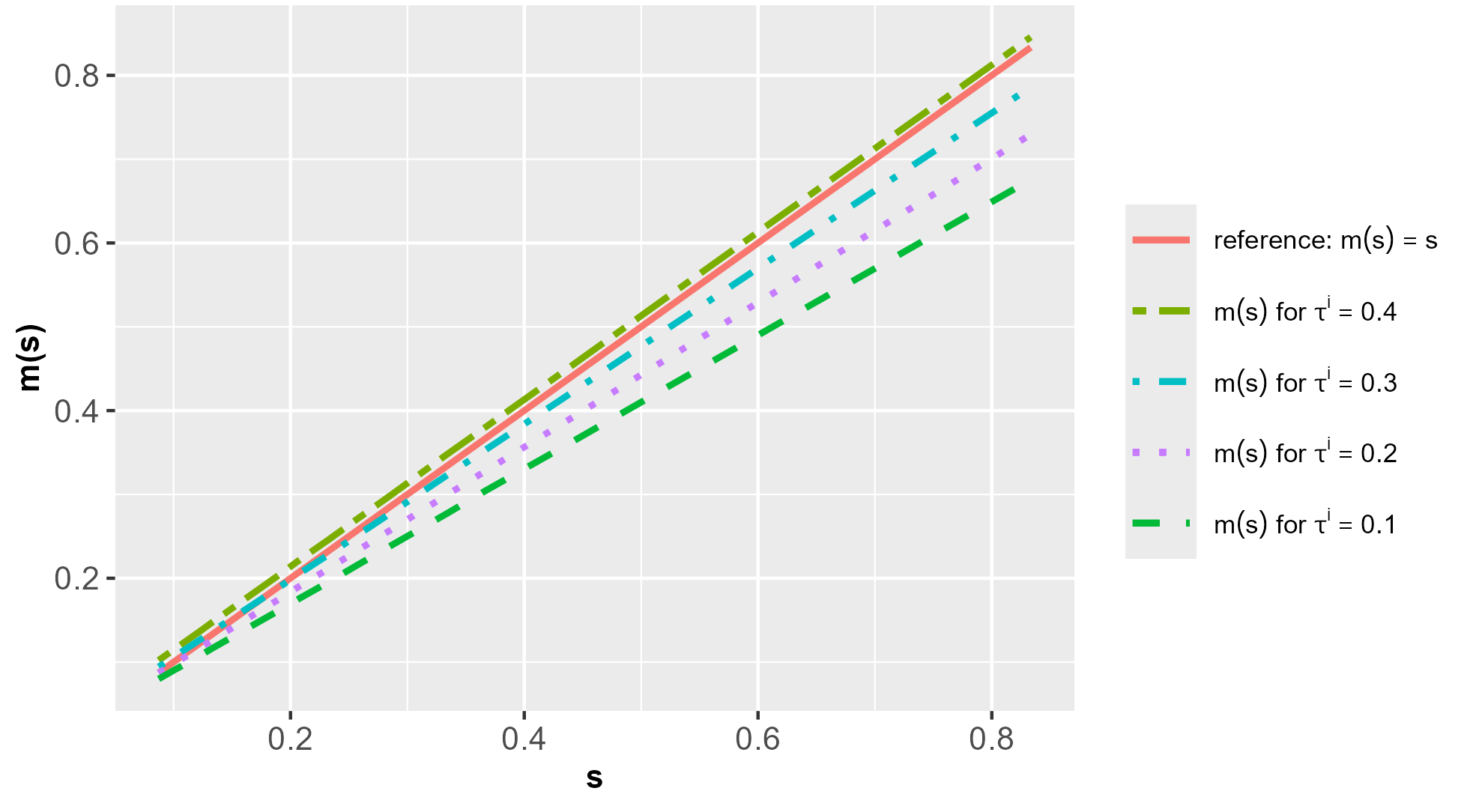}
        \caption{$m(s)$ vs $s$ for real losses}
    \end{subfigure}%
    \begin{subfigure}[b]{0.5\textwidth}
        \centering
        \includegraphics[width=\linewidth]{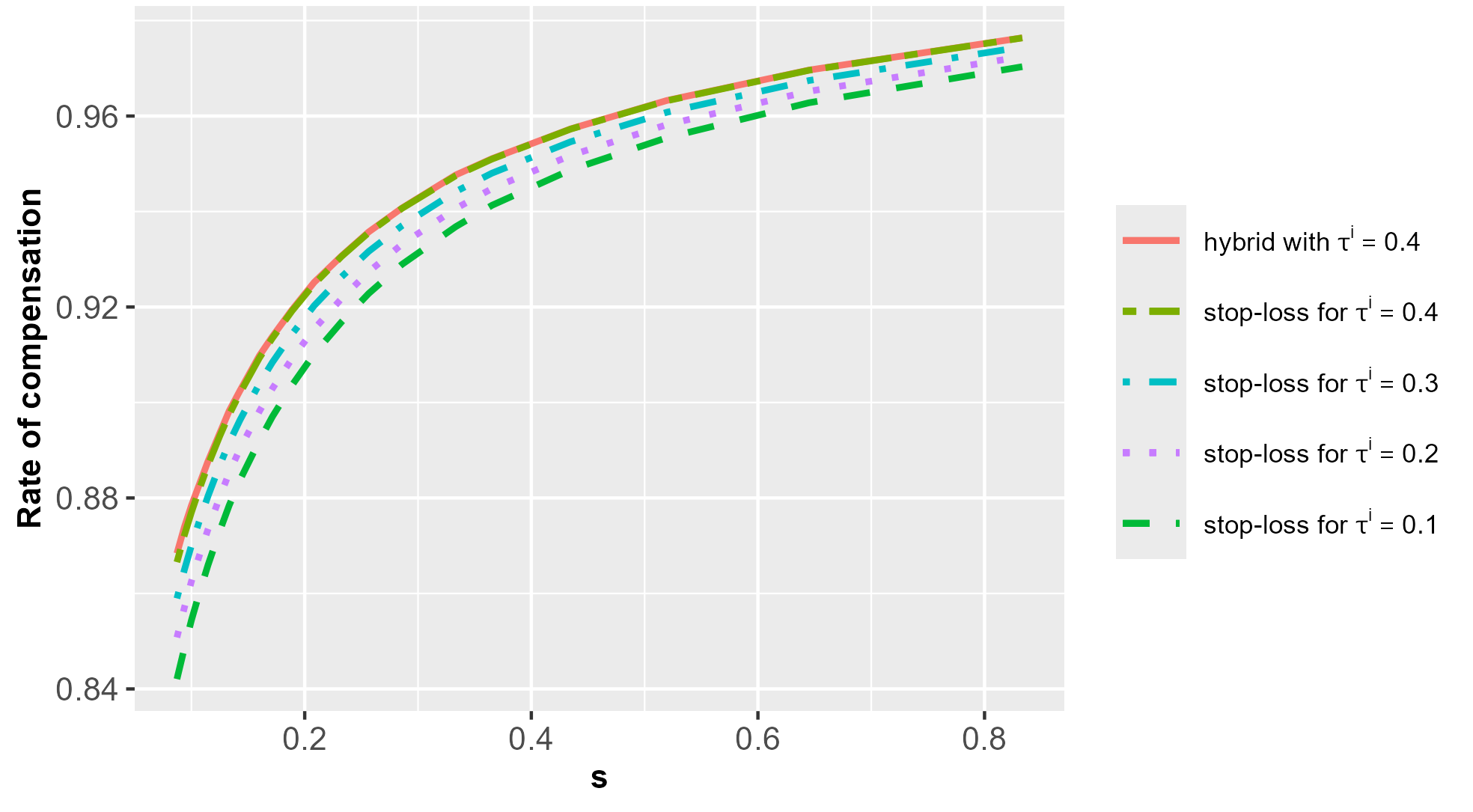}
        \caption{$\mathbf{E}[\frac{X}{Y}]$ vs $s$ for real losses}
    \end{subfigure}
    \\[5pt]
    \caption{Comparative analysis of expected rates of compensation of the proposed hybrid insurance contract and the classical capped indemnity contract. The results are presented for real and simulated data and for various values of $\tau^i$, the loading factor of index insurance. Note that $\tau^t$ is set at 40\%.}
    \label{fig:comp_analysis}
\end{figure}

Panels (a) and (c) of Figure \ref{fig:comp_analysis} show that for large values of $s$ (extreme losses) in the hybrid contract, the cap $m(s)$ of a capped indemnity contract of identical price is lower than $s$. Also, this value of $m(s)$ decreases as the loading of the index part of the proposed hybrid contract decreases. This highlights an additional advantage of the proposed hybrid contract in terms of quality of coverage. Indeed, these results are presented in the form of expected rates of compensation in panels (b) and (d) of the same figure. These panels show that the expected rates of compensation of the capped indemnity contract decrease as the price of the index part of the hybrid contract decreases for a given value of $s$. This means that the advantage of the lower cost of index insurance makes the proposed hybrid contract more competitive and more suitable for policyholders than a capped indemnity contract of equal price. Moreover, further work on reducing the basis risk in the index part of this product could greatly improve this enhanced coverage provided by index insurance. 

Also, recall that in the payout functions of the hybrid insurance contract presented in Sections \ref{sec:simul_set} and \ref{sec:real_set}, a condition is imposed such that compensation for losses above the threshold $s$ is never less than $s$. This hybrid contract therefore has an implicit cap condition in addition to using index values to provide suitable compensation above $s$. This provides another explanation for why the hybrid contract compensates more than a classical capped indemnity contract of equal price.

\section{Conclusion}

In this paper, we studied the conception of an insurance product that mixes traditional insurance and an index-based part. The idea is to improve the ability to cover risks that are heavy tail, instead of applying a fixed limit of compensation. The use of covariates is used to provide an indemnity which is as close as possible to the needs of the policyholder, with a fast computation of the compensation which does not require expert analysis. The calibration of the index-based part of the product is based on a specific metric. Moreover, we provided an approximation of this criterion that allows to enhance the calibration part via the use of additional data on the covariates available after a claim (without the need of getting additional data on the economic losses).

The criterion that we use is of course dependent on assumptions on the behavior of the policyholders, and the question of the choice of the function $L$ used for the optimization is a question that should be investigated from data on the behavior of the customers. This important question is beyond the scope of this paper. Regarding the behavior of the policyholder, we here consider that the only reason for preferring a reduced cover is a lower price. In practice, this may not be the sole factor: a faster compensation could be preferred, even though this compensation would be lower, since fast compensation may help the policyholder to repair faster, which is an important aspect of resilience. Moreover, even from a price prospective, the reduction of loading factors on the premiums of index insurance due to simplified claim management provides another level of optimization. This important effect of time is an aspect of the problem that could be investigated as an extension of the present work.

Another potential extension would concern the case where $\mathbf{W}$ is high-dimensional. In this case, the so-called "curse of dimensionality" may limit the quality of the index calibration, and/or could lead to overfitting due to the inability to understand the complexity of the impact of the covariates on the tail of the distribution. The study of the adaptation of tail index regression methods (see e.g. \cite{velthoen2023gradient}, \cite{rai2024fast}) should be considered to improve the method.

\section{Appendix}
\label{sec:appendix}

The Appendix section is organized as follows. In section \ref{sec:proof1}, we provide the proof of the key approximation result, that is Theorem \ref{sec:proof1}, followed by section \ref{sec:prooftrigger} which provides a technical lemma related to the analysis of the trigger mechanism. Section \ref{sec:proofestim} is related to the proof of Theorem \ref{th_theta} on the convergence rate of the index optimization procedure, which relies on a technical Lemma presented in section \ref{sec:technical}. Additional empirical results related to the real data analysis are provided in section \ref{sec:addresults}. Equations \ref{eq:summary_imp} summarize the key notations used. They are provided to facilitate the understanding of the results and the proofs.

\begin{equation}\label{eq:summary_imp}
    \begin{aligned}
        \mathfrak{L}_\theta(s)&=E\left[L\left(\frac{X_{\theta}(s)}{Y}-f(\pi_{\theta}(s))\right)\right] \text{ and } \tilde{\theta}(s)=\arg \max_{\theta\in \Theta}\mathfrak{L}_\theta(s)\\
        \mathfrak{L}^*_{\theta}(s)&=\mathfrak{L}_{\theta}(s)/(1+\frak{R}_{\theta}(s))\\ &= E\left[\Psi(\pi_{\theta}(s),\phi_{\theta}(\mathbf{W});S(s|\mathbf{W}),\gamma(\mathbf{W}))\right] \text{ and } \tilde{\theta}^*(s)=\arg \max_{\theta\in \Theta}\mathfrak{L}^*_{\theta}(s)\\
        \hat{\mathfrak{L}}^*_{\theta}(s)&=\frac{1}{m}\sum_{j=1}^m \Psi\left(\hat{\pi}_{\theta}(s),\phi_{\theta}(\mathbf{W}_j);\hat{S}(s|\mathbf{W}_j),\hat{\gamma}(\mathbf{W}_j)\right) \text{ and } \hat{\tilde{\theta}}^*(s)=\arg \max_{\theta \in \Theta}\hat{\mathfrak{L}}^*_{\theta}(s)
    \end{aligned}
\end{equation}

\subsection{Proof of Theorem \ref{th_approx}}
\label{sec:proof1}

Let 
$$\mathfrak{L}_{\theta}(s|\mathbf{w})=E\left[L\left(\frac{X_{\theta}(s)}{Y}-f(\pi_{\theta}(s))\right)|\mathbf{W}=\mathbf{w}\right].$$

Let us introduce $$\tilde{X}_{\theta}(s)=Y\mathbf{1}_{Y\leq s}+s\phi_{\theta}(\mathbf{W})\mathbf{1}_{Y>s}.$$ $\tilde{X}_{\theta}(s)$ is the "ideal" hybrid product that we would use if $Y$ could be known just after the claim. Then, we define $\tilde{\mathfrak{L}}_{\theta}(s|\mathbf{w})=E\left[\left(\frac{\tilde{X}_{\theta}(s)}{Y}-f({\pi}_{\theta}(s))\right)|\mathbf{W}=\mathbf{w}\right].$

We have
\begin{eqnarray}{\tilde{\mathfrak{L}}_{\theta}(s|\mathbf{w})}&= & E\left[L\left(\mathbf{1}_{Y\leq s}+\frac{s\phi_{\theta}(\mathbf{w})}{Y}\mathbf{1}_{Y>s}-f(\pi_{\theta}(s))\right)|\mathbf{W}=\mathbf{w}\right] \nonumber \\
&=& \mathbb{P}(Y\leq s|\mathbf{W}=\mathbf{w})L(1-f(\pi_{\theta}(s)))-\int_s^{\infty}L\left(\frac{s\phi_{\theta}(\mathbf{w})}{t}-f(\pi_{\theta}(s))\right)dS_Y(t|\mathbf{w}) \nonumber \\
&=& L(1-f(\pi_{\theta}(s)))\left\{1-S(s|\mathbf{w})\left[1-\frac{L\left(\phi_{\theta}(\mathbf{w})-f(\pi_{\theta}(s)\right)}{L\left(1-f(\pi_{\theta}(s))\right)}+\mathcal{L}_{\theta,1}(s|\mathbf{w})\right]\right\}, \label{d1}
\end{eqnarray}
with
\begin{eqnarray*}\mathcal{L}_{\theta,1}(s|\mathbf{w}) &=& s\phi_{\theta}(\mathbf{w}) \int_s^{\infty} \frac{L'\left(\frac{s\phi_{\theta}(\mathbf{w})}{t}-f(\pi_{\theta}(s))\right)}{L\left(1-f(\pi_{\theta}(s))\right)}\frac{S(t|\mathbf{w})dt}{S(s|\mathbf{w})t^2} \\
&=& \phi_{\theta}(\mathbf{w}) \int_{1}^{\infty} \frac{L'\left(\frac{\phi_{\theta}(\mathbf{w})}{u}-f(\pi_{\theta}(s))\right)}{L\left(1-f(\pi_{\theta}(s))\right)}\frac{l(su|\mathbf{w})}{l(s|\mathbf{w})}\frac{du}{u^{2+\frac{1}{\gamma(\mathbf{w})}}}.
\end{eqnarray*}

From Assumption \ref{a0},
\begin{equation}\frac{L\left(\phi_{\theta}(\mathbf{w})-f(\pi_{\theta}(s))\right)}{L\left(1-f(\pi_{\theta}(s))\right)}=\varphi_0(\phi_{\theta}(\mathbf{w}))(1+R_{\theta,0}(s,\mathbf{w}))\label{d2}\end{equation}
with $\sup_{\theta, \mathbf{w}}|R_{\theta,0}(s,\mathbf{w})|\rightarrow 0,$ when $s$ tends to infinity.

We have,
\begin{eqnarray*}\mathcal{L}_{\theta,1}(s|\mathbf{w})) &=& \phi_{\theta}(\mathbf{w})\int_{1}^{\infty}\varphi_1(\phi_{\theta}(\mathbf{w})/u)\frac{l(su|\mathbf{w})du}{l(s|\mathbf{w})u^{2+\frac{1}{\gamma(\mathbf{w})}}}(1+R_{\theta,1}(s,\mathbf{w})), 
\end{eqnarray*}
where, from Assumption \ref{a-1},

$$|R_{\theta,1}(s,\mathbf{w})|\leq \frac{C}{c}\left\{\sup_{t}\frac{\left|\frac{L'\left(t-f(\pi_{\theta}(s))\right)}{L\left(1-f(\pi_{\theta}(s))\right)}-\varphi_1(t)\right|}{\varphi_1(t)}\right\}\int_1^{\infty} \frac{l^{0}(su)}{l^{0}(s)}\frac{du}{u^2}.$$

Since $l^{0}$ is slow-varying, dominated convergence yields
$$\int_1^{\infty} \frac{l^{0}(su)}{l^{0}(s)}\frac{du}{u^2}\rightarrow \int_1^{\infty}\frac{du}{u^2},$$

when $s$ tends to infinity. Then, from Assumption \ref{a1}, we get \begin{equation}
\sup_{\theta,\mathbf{w}}|R_{\theta,1}(s,\mathbf{w})|=o(1).\label{d3}\end{equation}

Next, let
$$\mathcal{L}_{\theta,2}(s|\mathbf{w})=\phi_{\theta}(\mathbf{w})\int_{1}^{\infty}\varphi_1(\phi_{\theta}(\mathbf{w})/u)\frac{l(su|\mathbf{w})du}{l(s|\mathbf{w})u^{2+\frac{1}{\gamma(\mathbf{w})}}}.$$
We have
\begin{eqnarray*}\mathcal{L}_{\theta,2}(s|\mathbf{w}) &=& \phi_{\theta}(\mathbf{w})\int_{1}^{\infty}\varphi_1(\phi_{\theta}(\mathbf{w})/u)\frac{du}{u^{2+\frac{1}{\gamma(\mathbf{w})}}}+R_{\theta,2}(s,\mathbf{w}) \\
&=& \left\{\phi_{\theta}(\mathbf{w})^{-1/\gamma(\mathbf{w})}\int_0^{\phi_{\theta}(\mathbf{w})}v^{1/\gamma(\mathbf{w})}\varphi_1(v)dv\right\}(1+R_{\theta,2}(s,\mathbf{w})) \\
&=& \Phi_1(\phi_{\theta}(\mathbf{w}),\gamma(\mathbf{w})),
\end{eqnarray*}
with
$$R_{\theta,2}(s,\mathbf{w})=\phi_{\theta}(\mathbf{w})\int_{1}^{\infty}\varphi_1(\phi_{\theta}(\mathbf{w})/u)\left[\frac{l(su|\mathbf{w})}{l(s|\mathbf{w})}-1\right]\frac{du}{u^{2+\frac{1}{\gamma(\mathbf{w})}}}.$$
We have
$$R_{\theta,2}(s,\mathbf{w})\leq |\Phi_1(\phi_{\theta}(\mathbf{w}),\gamma(\mathbf{w}))|\sup_{\mathbf{w}}\left|\frac{l(su|\mathbf{w})}{l(s|\mathbf{w})}-1\right|,$$
and, from Assumption \ref{a-1},
\begin{equation}\sup_{\theta,\mathbf{w}}\left|\frac{R_{\theta,2}(s,\mathbf{w})}{\Phi_1(\phi_{\theta}(\mathbf{w}),\gamma(\mathbf{w}))}\right|\rightarrow_{s\rightarrow \infty} 0.\label{d4}\end{equation}

Gathering (\ref{d1}) to (\ref{d4}), we get
$$\tilde{\mathfrak{L}}_{\theta}(s|\mathbf{w})=L(1-f(\pi_{\theta}(s)))\left\{1-S(s|\mathbf{w})\left[1-\varphi_0(\phi_{\theta}(\mathbf{w}))+\Phi_1(\phi_{\theta}(\mathbf{w})\right]\right\}(1+R_{\theta}(s,\mathbf{w})),$$
with $\sup_{\theta,w}|R(s,\mathbf{w})|\rightarrow 0$ when $s$ tends to infinity.

The result then follows from Lemma \ref{lemma_1} below (whose proof is given in section \ref{sec:prooftrigger}), since $\mathfrak{L}_{\theta}(s)=E\left[\mathfrak{L}_{\theta}(s|\mathbf{W})\right].$

\begin{lemma}
\label{lemma_1}
Under the assumptions of Theorem \ref{th_approx},
$$\left|\mathfrak{L}_{\theta}(s|\mathbf{w})-\tilde{\mathfrak{L}}_{\theta}(s|\mathbf{w})\right|\leq \sup_{x}|L'(x)|\times \max\left(\frac{s\Phi(\mathbf{w})}{Y}-1,1\right)\mathbf{1}_{F(s)}.$$
\end{lemma}

\subsection{Proof of Lemma \ref{lemma_1}}
\label{sec:prooftrigger}

We have
$$\left|L\left(\frac{X_{\theta}(s)}{Y}-f(\pi_{\theta}(s))\right)-L\left(\frac{\tilde{X}_{\theta}(s)}{Y}-f(\pi_{\theta}(s))\right)\right|\leq \sup_{x}|L'(x)|\times \left|\frac{s\phi_{\theta}(\mathbf{W})}{Y}-1\right|\mathbf{1}_{F(s)}.$$ If $\left|\frac{s\phi_{\theta}(\mathbf{W})}{Y}-1\right|=\frac{s\phi_{\theta}(\mathbf{W})}{Y}-1,$ then it is bounded by $s\Phi(\mathbf{W})/Y-1.$ In the other case, it is bounded by 1, and the result follows by taking the conditional expectation.

\subsection{Proof of Theorem \ref{th_theta}}
\label{sec:proofestim}

Let us denote
$$\bar{\mathfrak{L}}^*_{\theta}(s)=\frac{1}{m}\sum_{j=1}^m \Psi\left(\pi_{\theta}(s),\phi_{\theta}(\mathbf{W}_{j});S(s|\mathbf{W}_{j}),\gamma(\mathbf{W}_{j})\right).$$ This function of $\theta$ is the one we could compute if we knew exactly $S$ and $\gamma.$ The proof Theorem \ref{th_theta} is essentially based on the fact that the difference between $\bar{\mathfrak{L}}_{\theta}^*(s)$ and $\hat{\mathfrak{L}}_{\theta}^*(s)$ is asymptotically negligible. Thus, $\bar{\tilde{\theta}}^*(s)=\arg \max_{\theta\in \Theta}\bar{\mathfrak{L}}_{\theta}^*(s)$ and $\hat{\tilde{\theta}}^*(s)$ are asymptotically equivalent.

The result of the Theorem is shown in two steps. First, we show the consistency of $\hat{\tilde{\theta}}^*(s)$ (Proposition \ref{prop_prel}), then the convergence rate is derived in a second step.

\textbf{Step 1: consistency.}

\begin{proposition}
\label{prop_prel}
Under the assumptions of Theorem \ref{th_theta},
$$\hat{\tilde{\theta}}^*(s)-\tilde{\theta}^*(s)=o_P(1).$$
\end{proposition}

\begin{proof}
To prove the consistency, it suffices to show that
$$\sup_{\theta}|\hat{\mathfrak{L}}_{\theta}^*(s)-\mathfrak{L}^*_{\theta}(s)|=o_P(1),$$ since $\tilde{\theta}^*(s)$ is the unique maximizer of $\mathfrak{L}^*_{\theta}(s)$ and is in the interior part of $\Theta$ (see for example Theorem 5.7 in \cite{van2000asymptotic}).

We have \begin{equation}\sup_{\theta \in \Theta} \left|\hat{\mathfrak{L}}_{\theta}^*(s)-\mathfrak{L}^*_{\theta}(s)\right|\leq \sup_{\theta \in \Theta} \left|\hat{\mathfrak{L}}_{\theta}^*(s)-\bar{\mathfrak{L}}_{\theta}^*(s)\right|+ \sup_{\theta \in \Theta} \left|\bar{\mathfrak{L}}_{\theta}^*(s)-\mathfrak{L}^*_{\theta}(s)\right|.\label{eqq0}\end{equation}
We have
\begin{eqnarray*}
\left|\Psi\left(\pi_{\theta}(s),\phi_{\theta}(\mathbf{w});S(s|\mathbf{w}),\gamma(\mathbf{w})\right)-\Psi\left(\pi_{\theta'}(s),\phi_{\theta'}(\mathbf{w});S(s|\mathbf{w}),\gamma(\mathbf{w})\right)\right| &\leq & \|\theta-\theta'\|L(\mathbf{w}),
\end{eqnarray*}
with
\begin{eqnarray*}
L(\mathbf{w}) &=& \sup_{\theta\in \Theta}\left|-f'(\pi_{\theta}(s))\frac{\partial \pi_{\theta}(s)}{\partial \theta}L'(1-f(\pi_{\theta}(s)))(1-S(s|\mathbf{w}))\Phi_0(\phi_{\theta}(\mathbf{w}),\gamma(\mathbf{w}))\right| \\
&&+\sup_{\theta\in \Theta}\left|\frac{\partial \phi_{\theta}(\mathbf{w})}{\partial \theta}L\left(1-f(\pi_{\theta}(s))\right)(1-S(s|\mathbf{w}))\partial_1\Phi_0(\phi_{\theta}(\mathbf{w}),\gamma(\mathbf{w}))\right| \\
&&\leq C\Lambda(\mathbf{w}),
\end{eqnarray*}
from Assumption \ref{a_uf}, for some constant $C.$

Therefore, from Example 19.7 in \cite{van2000asymptotic}, we get 
\begin{equation}
\label{eqq1}\sup_{\theta \in \Theta} \left|\bar{\mathfrak{L}}_{\theta}^*(s)-\mathfrak{L}^*_{\theta}(s)\right|=o_P(1).
\end{equation}

Next,
\begin{eqnarray*}\hat{\mathfrak{L}}_{\theta}^*(s)-\bar{\mathfrak{L}}_{\theta}^*(s) &=& \frac{1}{m}\sum_{j=1}^m \left\{\frac{S(s|\mathbf{W}_{j})-\hat{S}(s|\mathbf{W}_{j})}{1-S(s|\mathbf{W}_{j})} + \int_0^{\phi_{\theta}(\mathbf{W}_{j})} \left(v^{1/\hat{\gamma}(\mathbf{W}_{j})}-v^{1/\gamma(\mathbf{W}_{j})}\right)\varphi_1(v)dv\right\}\\
&& \times \Psi\left(\pi_{\theta}(s),\phi_{\theta}(\mathbf{W}_{j});S(s|\mathbf{W}_{j}),\gamma(\mathbf{W}_{j})\right).
\end{eqnarray*}
We get, for some constant $C_0,$
\begin{eqnarray*}\left|\hat{\mathfrak{L}}_{\theta}^*(s)-\bar{\mathfrak{L}}_{\theta}^*(s) \right| &\leq & C_0
\sup_{\mathbf{w}\in \mathcal{W}}\left|\frac{S(s|\mathbf{w})-\hat{S}(s|\mathbf{w})}{1-S(s|\mathbf{w})}\right|\times \sup_{\mathbf{w}\in \mathcal{W}}\left|\frac{1}{\hat{\gamma}(\mathbf{w})}-\frac{1}{{\gamma}(\mathbf{w})}\right|\int_0^{\infty} v^{1/\gamma_-} \varphi_1(v)dv,
\end{eqnarray*}
since $\Psi$ is bounded. Assumption \ref{a_donsker} and Assumption \ref{a_uf} lead to
\begin{equation}
\label{eqq2}\sup_{\theta \in \Theta} \left|\hat{\mathfrak{L}}_{\theta}^*(s)-\bar{\mathfrak{L}}_{\theta}^*(s)\right|=o_P(1).
\end{equation}

Merging (\ref{eqq0}) to (\ref{eqq2}), we get the result.
\end{proof}

\textbf{Step 2: Rate of consistency.}

To shorten the notation, $\hat{\tilde{\theta}}^*$ (resp. $\bar{\tilde{\theta}}^*,$ $\tilde{\theta}^*$) will denote $\hat{\tilde{\theta}}^*(s)$ (resp. $\bar{\tilde{\theta}}^*(s),$ $\tilde{\theta}^*(s)$) in this section. Let $\nabla_{\theta}$ denote the gradient vector of a function with respect to $\theta.$ By definition,
$$\nabla_{\theta} \hat{\mathfrak{L}}_{\hat{\tilde{\theta}}^*}^*(s)=0.$$
From a Taylor expansion and the consistency of $\hat{\tilde{\theta}}^*,$
\begin{equation}\label{it0} 0=\nabla_{\theta} \hat{\mathfrak{L}}^*_{\tilde{\theta}^*}(s)+[{\tilde{\theta}}^*-\hat{\tilde{\theta}}^*][\nabla_{\theta}^2 \hat{\mathfrak{L}}^*_{\tilde{\theta}^*}(s)+o_P(1)].\end{equation}
From the consistency of $\hat{S}$ and $\hat{\gamma}$ in Assumption \ref{a_donsker} and the law of large numbers, we get
$$\nabla_{\theta}^2 \hat{\mathfrak{L}}^*_{\tilde{\theta}^*}(s)=\nabla_{\theta}^2 \bar{\mathfrak{L}}^*_{\tilde{\theta}^*}(s)+o_P(1)= \nabla_{\theta}^2 {\mathfrak{L}}^*_{\tilde{\theta}^*}(s)+o_P(1)=\Sigma(s)+o_P(1).$$

Next,
\begin{eqnarray*}\nabla_{\theta} \hat{\mathfrak{L}}^*_{\tilde{\theta}^*}(s)&=&\nabla_{\theta}\bar{\mathfrak{L}}^*_{\tilde{\theta}^*}(s)+\frac{1}{m}\sum_{j=1}^m \left\{\nabla_{\theta}\Psi\left(\pi_{\tilde{\theta}^*}(s),\phi_{\tilde{\theta}^*}(\mathbf{W}_{j});\hat{S}(s|\mathbf{W}_{j}),\hat{\gamma}(\mathbf{W}_{j})\right)\right. \\ && \left. -\nabla_{\theta}\Psi\left(\pi_{\tilde{\theta}^*}(s),\phi_{\tilde{\theta}^*}(\mathbf{W}_{j});{S}(s|\mathbf{W}_{j}),{\gamma}(\mathbf{W}_{j})\right)\right\}.
\end{eqnarray*}
From Assumption \ref{a_donskerbis} and Lemma \ref{lemma_donsker}, we get, from the asymptotic equicontinuity property of Donsker classes (see section 2.8.2 in \cite{van1996weak}),
\begin{eqnarray*}
\frac{1}{m}\sum_{j=1}^m \left\{\nabla_{\theta}\Psi\left(\pi_{\tilde{\theta}^*}(s),\phi_{\tilde{\theta}^*}(\mathbf{W}_{j});\hat{S}(s|\mathbf{W}_{j}),\hat{\gamma}(\mathbf{W}_{j})\right) \right. \\ \left. -\nabla_{\theta}\Psi\left(\pi_{\tilde{\theta}^*}(s),\phi_{\tilde{\theta}^*}(\mathbf{W}_{j});{S}(s|\mathbf{W}_{j}),{\gamma}(\mathbf{W}_{j})\right)\right\} \\
=\int \left\{\nabla_{\theta}\Psi\left(\pi_{\tilde{\theta}^*}(s),\phi_{\tilde{\theta}^*}(\mathbf{w});\hat{S}(s|\mathbf{w}),\hat{\gamma}(\mathbf{w})\right)\right\}d\mathbb{P}(\mathbf{w})+o_P(m^{-1/2})=B_n+o_P(m^{-1/2}).
\end{eqnarray*}
Back to (\ref{it0}), we get
$${\tilde{\theta}}^*-\hat{\tilde{\theta}}^*=\Sigma^{-1}\left\{\frac{1}{m}\sum_{j=1}^m \nabla_{\theta}\Psi\left(\pi_{\tilde{\theta}^*}(s),\phi_{\tilde{\theta}^*}(\mathbf{W}_j);S(s|\mathbf{W}_j),\gamma(\mathbf{W}_j)\right)\right\}+O_P(B_n)+o_P(m^{-1/2}).$$

\subsection{Technical lemma on Donsker classes}
\label{sec:technical}

\begin{lemma}
\label{lemma_donsker}
Let $$\mathcal{A}=\left\{\mathbf{w}\rightarrow \nabla_{\theta}\Psi\left(\pi_{\theta}(s),\phi_{\theta}(\mathbf{w});h_1(s,\mathbf{w}),h_2(\mathbf{w})\right):h_1\in \mathcal{H}_1,h_2\in \mathcal{H}_2,\theta \in \Theta\right\}.$$
Then $\mathcal{A}$ is Donsker under Assumption \ref{a_donskerbis}.
\end{lemma}

\begin{proof}
Let $$\mathcal{A}_3=\left\{\mathbf{w}\rightarrow \frac{\partial\phi_{\theta}(\mathbf{w})}{\partial \theta}:\theta\in \Theta\right\}.$$
We have
\begin{eqnarray*}\nabla_{\theta}\Psi\left(\pi_{\theta}(s),\phi_{\theta}(\mathbf{w});h_1(s,\mathbf{w}),h_2(\mathbf{w})\right) = \\
-f'(\pi_{\theta}(s))\frac{\partial \pi_{\theta}(s)}{\partial \theta}L'(1-f(\pi_{\theta}(s)))(1-h_1(s,\mathbf{w})) \Phi_0(\phi_{\theta}(\mathbf{w}),h_2(\mathbf{w}))
\\
+\frac{\partial \phi_{\theta}(\mathbf{w})}{\partial \theta}L\left(1-f(\pi_{\theta}(s))\right)(1-h_1(s,\mathbf{w}))\partial_1\Phi_0(\phi_{\theta}(\mathbf{w}),h_2(\mathbf{w})).
\end{eqnarray*}
The classes of functions $\mathcal{A}_1=\{\mathbf{w}\rightarrow \Phi_0(\phi_{\theta}(\mathbf{w}),h_2(\mathbf{w})):h_2\in \mathcal{H}_2\}$ and  $\mathcal{A}_2=\{\mathbf{w}\rightarrow \partial_1\Phi_0(\phi_{\theta}(\mathbf{w}),h_2(\mathbf{w})):h_2\in \mathcal{H}_2\}$ inherit the Donsker property of classes $\{\mathbf{w}\rightarrow \phi_{\theta}(\mathbf{w}):\theta\in \Theta\}$ and $\mathcal{H}_2$ from Theorem 2.10.6 in \cite{van1996weak}. Indeed,
\begin{eqnarray*}|\Phi_0(\phi_{\theta}(\mathbf{w}),h_2(\mathbf{w}))-\Phi_0(\phi_{\theta'}(\mathbf{w}),h'_2(\mathbf{w}))|&\leq &|\varphi_0(\phi_{\theta}(\mathbf{w}))-\varphi_0(\phi_{\theta'}(\mathbf{w}))|+\left|\int_{\phi_{\theta}(\mathbf{w})}^{\phi_{\theta'}(\mathbf{w})}v^{1/\gamma_-}\varphi_1(v)dv\right| \\
&&+\left|\frac{1}{h_2(\mathbf{w})}-\frac{1}{h'_2(\mathbf{w})}\right|\int_0^{\infty}v^{\alpha}\varphi_1(v)dv \\
&\leq &C\left(\left|\phi_{\theta}(\mathbf{w})-\phi_{\theta'}(\mathbf{w})\right|+\left|h_2(\mathbf{w})-h'_2(\mathbf{w})\right|\right),
\end{eqnarray*}
for some constant $C,$
and
\begin{eqnarray*}|\partial_1\Phi_0(\phi_{\theta}(\mathbf{w}),h_2(\mathbf{w}))-\partial_1\Phi_0(\phi_{\theta'}(\mathbf{w}),h'_2(\mathbf{w}))|\leq \\ |\varphi_0(\phi_{\theta}(\mathbf{w}))-\varphi_0(\phi_{\theta'}(\mathbf{w}))|\\
+\left|\phi_{\theta}(\mathbf{w})^{1/h_2(\mathbf{w})}\varphi_1(\phi_{\theta}(\mathbf{w})) -\phi_{\theta}(\mathbf{w})^{1/h'_2(\mathbf{w})}\varphi_1(\phi_{\theta'}(\mathbf{w}))\right|
\\
\leq C' \left(\left|\phi_{\theta}(\mathbf{w})-\phi_{\theta'}(\mathbf{w})\right|+\left|h_2(\mathbf{w})-h'_2(\mathbf{w})\right|\right),
\end{eqnarray*}
for a constant $C'.$ Then, the classes 
$(1-\mathcal{H}_1)\mathcal{A}_1$ and $\mathcal{A}_3(1-\mathcal{H}_1)\mathcal{A}_2$ are Donsker from the boundedness of $\mathcal{H}_1$ and $\mathcal{A}_3$ and Example 2.10.8 in \cite{van1996weak}. Note that the boundedness of $\mathcal{A}_3$ can be weakened by additional conditions on $\mathcal{A}_2$: the key issue is that the class $\mathcal{A}_3(1-\mathcal{H}_1)\mathcal{A}_2$ should be dominated by a $L^2$ random variable. Then, the class $\mathcal{A}$ is included in the addition of these two classes (Donsker from Example 2.10.7 in \cite{van1996weak}), up to multiplications by uniformly bounded quantities depending on $\theta$ (but not on $\mathbf{w}$). Theorem 2.10.6 in \cite{van1996weak} applies again, leading to the result. 
\end{proof}

\subsection{Additional information on the distribution of real data}
\label{sec:addresults}

\begin{figure}[!h]
    \centering
    \begin{subfigure}[b]{0.8\textwidth}
        \centering
        \includegraphics[width=\linewidth]{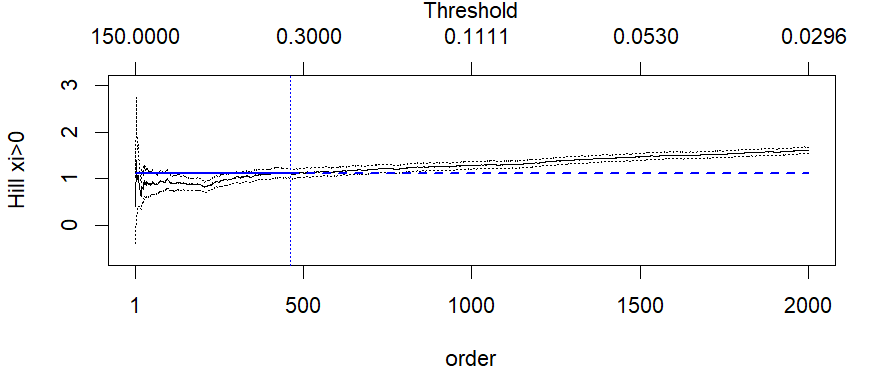}
        \caption{Hill plot}
    \end{subfigure}%
    \\[5pt]
    \begin{subfigure}[b]{0.8\textwidth}
        \centering
        \includegraphics[width=\linewidth]{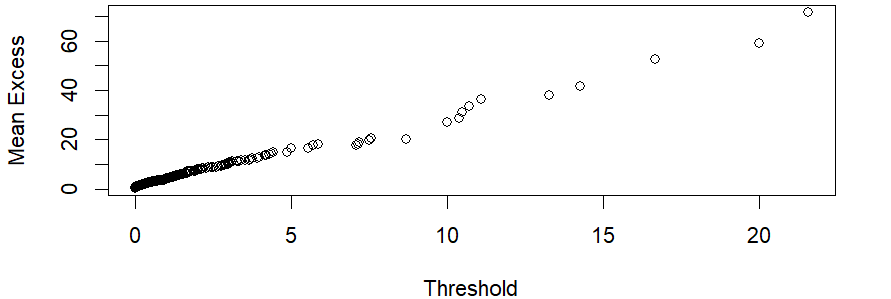}
        \caption{Mean excess plot}
    \end{subfigure}
    \\[5pt]
    \begin{subfigure}[b]{0.8\textwidth}
        \centering
        \includegraphics[width=\linewidth]{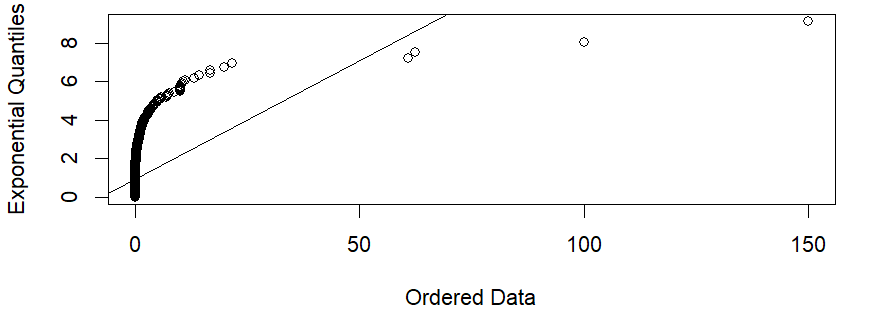}
        \caption{QQ plot}
    \end{subfigure}
    \\[5pt]
    \caption{Diagnosis plot for the tail of the distribution of tornado losses per unit area.}
    \label{fig:ev_analysis}
\end{figure}

\begin{figure}[!h]
    \centering
    \begin{subfigure}[b]{0.5\textwidth}
        \centering
        \includegraphics[width=\linewidth]{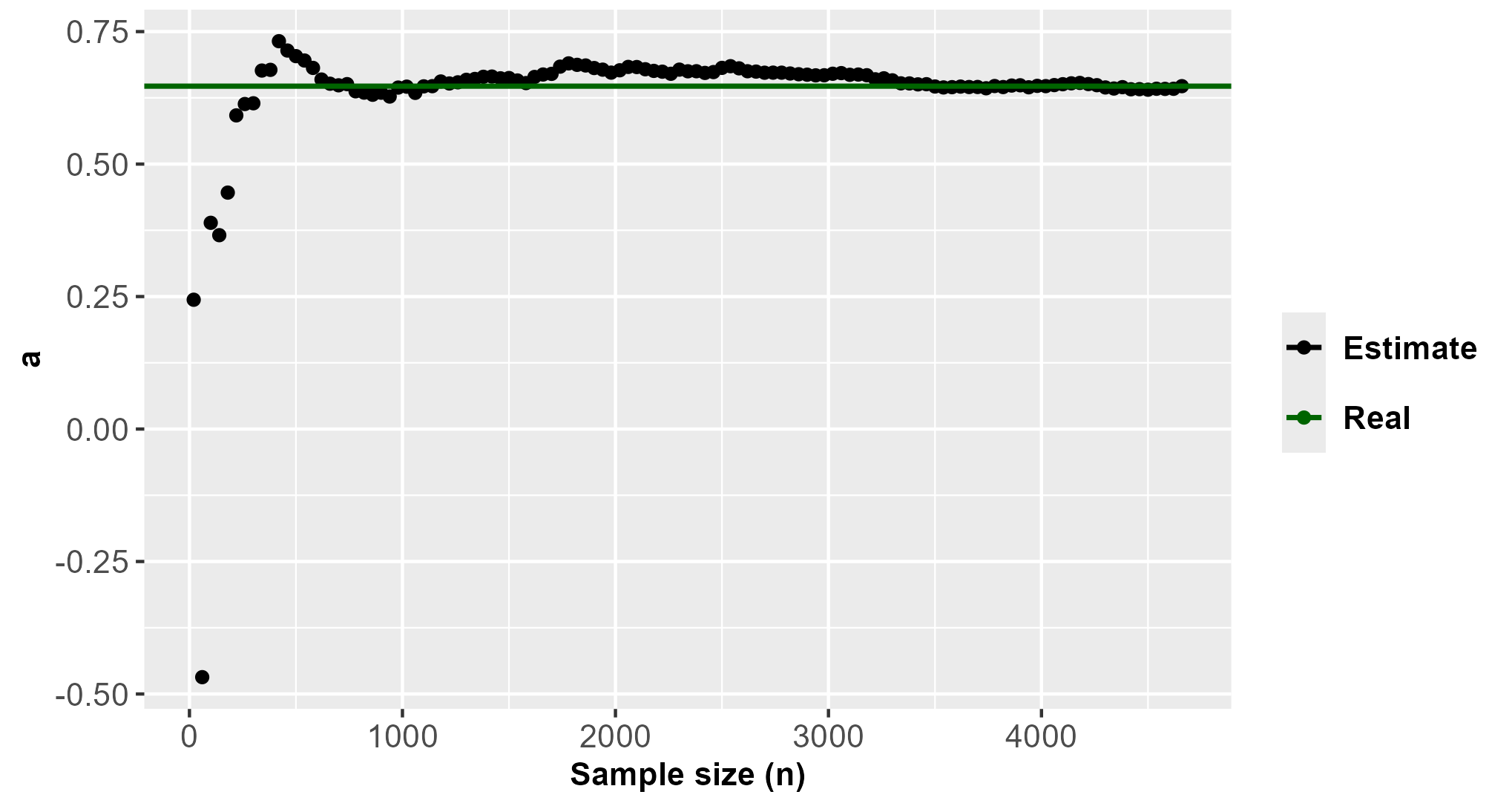}
        \caption{Estimates of the parameter $a$}
    \end{subfigure}%
    \begin{subfigure}[b]{0.5\textwidth}
        \centering
        \includegraphics[width=\linewidth]{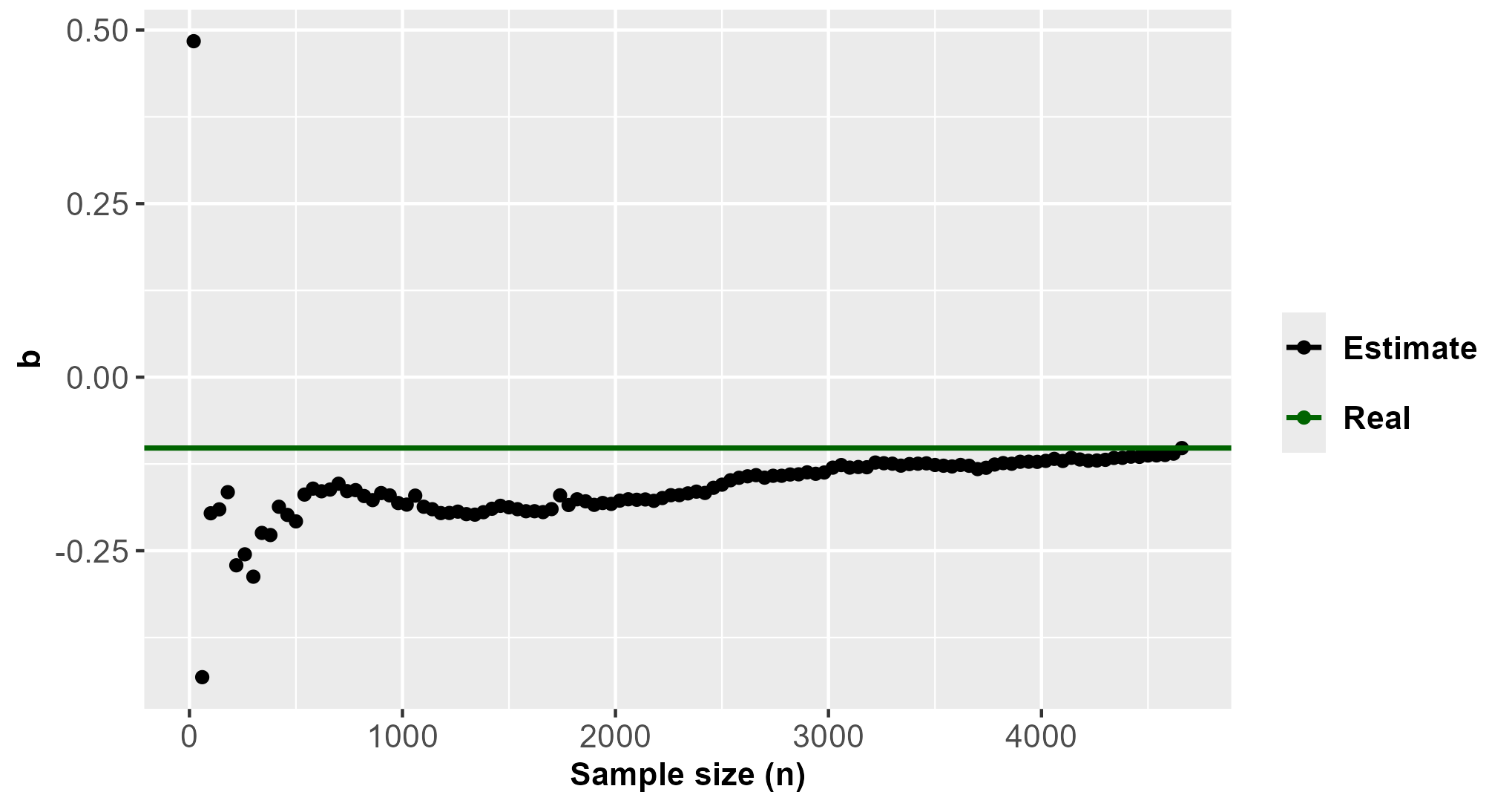}
        \caption{Estimates of the parameter $b$}
    \end{subfigure}
    \\[5pt]
    \begin{subfigure}[b]{0.5\textwidth}
        \centering
        \includegraphics[width=\linewidth]{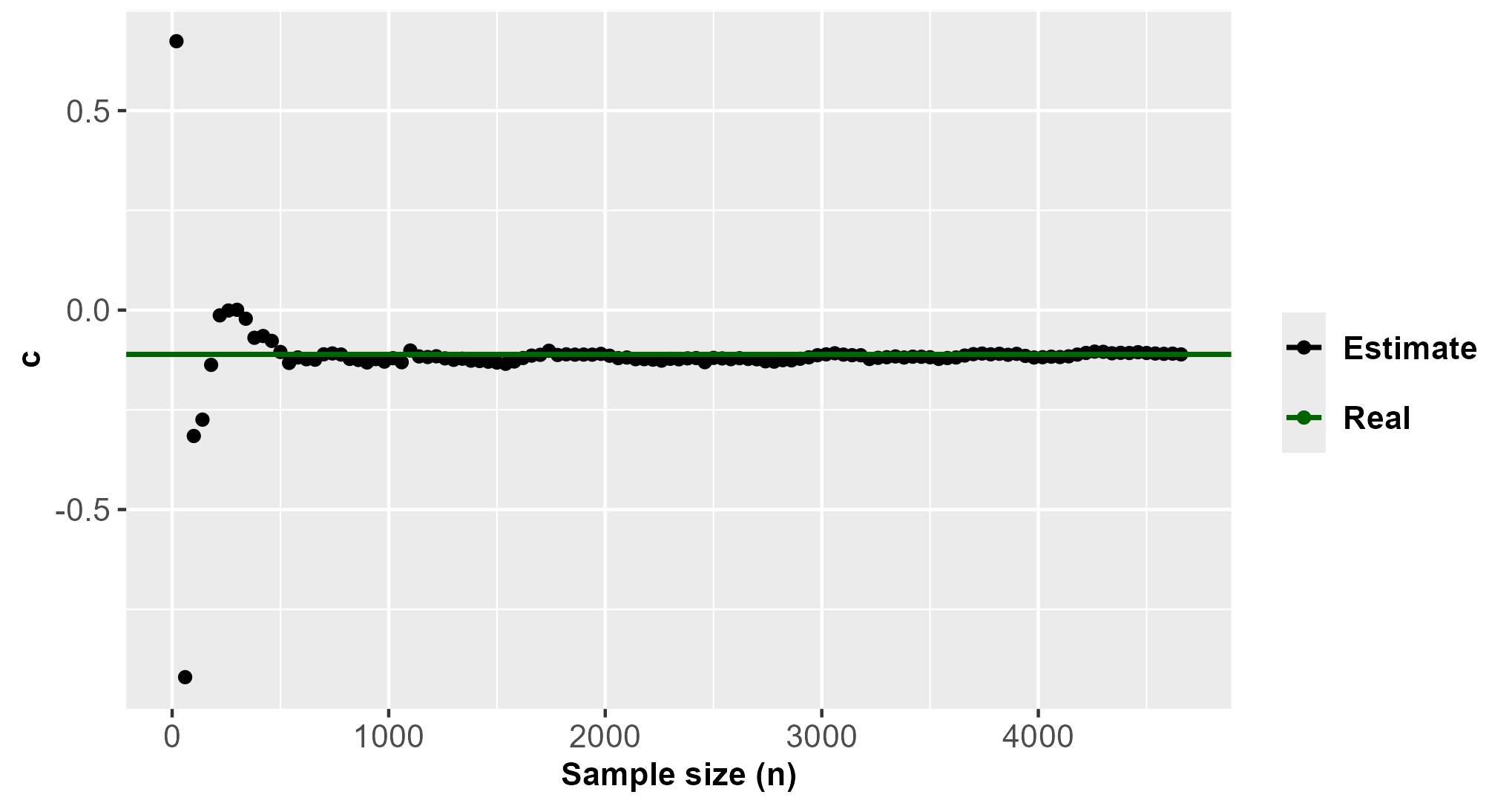}
        \caption{Estimates of the parameter $c$}
    \end{subfigure}
    \caption{Evolutions of the estimates of the parameters $a$ (panel (a)), $b$ (panel (b)) and $c$ (panel(c) of equation \ref{eq:gamma_real} with sample size.}
    \label{fig:param_real}
\end{figure}

\begin{table}[h!]
\centering
\begin{tabular}{lcccc}
\hline
\textbf{Parameter} & \textbf{Estimate} & \textbf{Std. Error} & \textbf{z value} & \textbf{Pr($>|z|$)} \\ \hline
\textbf{$\hat{a}$}     & 0.6470209 & 0.01854 & 34.898 & $ < 2e-16^{***}$ \\
\textbf{$\hat{b}$}     & -0.1020842  & 0.01859 & -5.492    & $3.96e-08^{***}$ \\
\textbf{$\hat{c}$}     & -0.1111014  & 0.01864 & -5.962    & $2.49e-09^{***}$ \\ 
\textbf{$\hat{d}$}     & -4.09608  & 0.03231 & -126.792    & $< 2e-16^{***}$ \\ 
\textbf{$\hat{e}$}     & 0.37194  & 0.03176 & 11.709    & $< 2e-16^{***}$ \\ 
\textbf{$\hat{f}$}     & 0.26602  & 0.03178 & 8.372    & $< 2e-16^{***}$ \\ \hline
\end{tabular}
\caption{Summary of the GPD shape and scale parameters fit on real data}
\label{tab:fitgpd}
\end{table}

\begin{table}[h!]
\centering
\begin{tabular}{lcccccc}
\hline
\multicolumn{7}{l}{\textbf{Model 1: $Y \sim 1$ (gpd fit with no covariates)}}   \\
\multicolumn{7}{l}{\textbf{Model 2: $Y \sim \mathbf{W}^{(1)} + \mathbf{W}^{(2)}$ (gpd fit with covariates)}} \\ \hline
\textbf{Model} & \textbf{Resid. Df} & \textbf{LogLik} & \textbf{AIC} & \textbf{Df} & \textbf{2 * LogLik Diff.} & \textbf{Pr(\textgreater{}Chi)} \\ \hline
\textbf{1}   & 9316  & 7658.7 & -15313.47 &   &        &                       \\
\textbf{2}   & 9312  & 7745.2 & -15478.32  & 4 & 172.85 & $<2.2e-16^{***}$ \\ \hline
\end{tabular}
\caption{Analysis of GPD fit deviance}
\label{tab:deviancegpd}
\end{table}

\newpage

\textbf{Acknowledgment:} Olivier Lopez acknowledges funding from the Excellence Chair CARE (Allianz, Ensae, Risk Fundation).

\bibliographystyle{abbrvnat}

\begin{thebibliography}{40}
\providecommand{\natexlab}[1]{#1}
\providecommand{\url}[1]{\texttt{#1}}
\expandafter\ifx\csname urlstyle\endcsname\relax
  \providecommand{\doi}[1]{doi: #1}\else
  \providecommand{\doi}{doi: \begingroup \urlstyle{rm}\Url}\fi

\bibitem[Abdi et~al.(2022)Abdi, Raffar, Zulkafli, Nurulhuda, Rehan, Muharam, Khosim, and Tangang]{abdi2022index}
M.~J. Abdi, N.~Raffar, Z.~Zulkafli, K.~Nurulhuda, B.~M. Rehan, F.~M. Muharam, N.~A. Khosim, and F.~Tangang.
\newblock Index-based insurance and hydroclimatic risk management in agriculture: A systematic review of index selection and yield-index modelling methods.
\newblock \emph{International Journal of Disaster Risk Reduction}, 67:\penalty0 102653, 2022.

\bibitem[Alderman and Haque(2007)]{alderman2007insurance}
H.~Alderman and T.~Haque.
\newblock \emph{Insurance against covariate shocks: The role of index-based insurance in social protection in low-income countries of Africa}.
\newblock Number~95. World Bank Publications, 2007.

\bibitem[Andr{\'e} et~al.(2013)Andr{\'e}, Monfort, Bouzit, and Vinchon]{andre2013contribution}
C.~Andr{\'e}, D.~Monfort, M.~Bouzit, and C.~Vinchon.
\newblock Contribution of insurance data to cost assessment of coastal flood damage to residential buildings: insights gained from johanna (2008) and xynthia (2010) storm events.
\newblock \emph{Natural Hazards and Earth System Sciences}, 13\penalty0 (8):\penalty0 2003--2012, 2013.

\bibitem[Barnett and Mahul(2007)]{barnett2007weather}
B.~J. Barnett and O.~Mahul.
\newblock Weather index insurance for agriculture and rural areas in lower-income countries.
\newblock \emph{American Journal of Agricultural Economics}, 89\penalty0 (5):\penalty0 1241--1247, 2007.

\bibitem[Beirlant et~al.(2004)Beirlant, Goegebeur, Segers, and Teugels]{beirlant}
J.~Beirlant, Y.~Goegebeur, J.~Segers, and J.~L. Teugels.
\newblock \emph{Statistics of extremes: theory and applications}, volume 558.
\newblock John Wiley \& Sons, 2004.

\bibitem[Biagini et~al.(2008)Biagini, Bregman, and Meyer-Brandis]{BIAGINI2008214}
F.~Biagini, Y.~Bregman, and T.~Meyer-Brandis.
\newblock Pricing of catastrophe insurance options written on a loss index with reestimation.
\newblock \emph{Insurance: Mathematics and Economics}, 43\penalty0 (2):\penalty0 214--222, 2008.
\newblock ISSN 0167-6687.
\newblock \doi{https://doi.org/10.1016/j.insmatheco.2008.05.016}.
\newblock URL \url{https://www.sciencedirect.com/science/article/pii/S0167668708000796}.

\bibitem[Braun et~al.(2023)Braun, Eling, and Jaenicke]{braun2023cyber}
A.~Braun, M.~Eling, and C.~Jaenicke.
\newblock Cyber insurance-linked securities.
\newblock \emph{ASTIN Bulletin: The Journal of the IAA}, 53\penalty0 (3):\penalty0 684--705, 2023.

\bibitem[Cairns and Boukfaoui(2021)]{Cairns18022021}
A.~J.~G. Cairns and G.~E. Boukfaoui.
\newblock Basis risk in index-based longevity hedges: A guide for longevity hedgers.
\newblock \emph{North American Actuarial Journal}, 25\penalty0 (sup1):\penalty0 S97--S118, 2021.
\newblock \doi{10.1080/10920277.2019.1651658}.
\newblock URL \url{https://doi.org/10.1080/10920277.2019.1651658}.

\bibitem[Carter et~al.(2017)Carter, de~Janvry, Sadoulet, and Sarris]{carter2017index}
M.~Carter, A.~de~Janvry, E.~Sadoulet, and A.~Sarris.
\newblock Index insurance for developing country agriculture: a reassessment.
\newblock \emph{Annual Review of Resource Economics}, 9\penalty0 (1):\penalty0 421--438, 2017.

\bibitem[Chantarat et~al.(2013)Chantarat, Mude, Barrett, and Carter]{chantarat2013designing}
S.~Chantarat, A.~G. Mude, C.~B. Barrett, and M.~R. Carter.
\newblock Designing index-based livestock insurance for managing asset risk in northern kenya.
\newblock \emph{Journal of Risk and Insurance}, 80\penalty0 (1):\penalty0 205--237, 2013.

\bibitem[Chen et~al.(2024)Chen, Lu, Zhang, and Zhu]{chen2024managing}
Z.~Chen, Y.~Lu, J.~Zhang, and W.~Zhu.
\newblock Managing weather risk with a neural network-based index insurance.
\newblock \emph{Management Science}, 70\penalty0 (7):\penalty0 4306--4327, 2024.

\bibitem[Clarke(2016)]{clarke2016theory}
D.~J. Clarke.
\newblock A theory of rational demand for index insurance.
\newblock \emph{American Economic Journal: Microeconomics}, 8\penalty0 (1):\penalty0 283--306, 2016.

\bibitem[Conradt et~al.(2015)Conradt, Finger, and Bokusheva]{conradt2015tailored}
S.~Conradt, R.~Finger, and R.~Bokusheva.
\newblock Tailored to the extremes: Quantile regression for index-based insurance contract design.
\newblock \emph{Agricultural economics}, 46\penalty0 (4):\penalty0 537--547, 2015.

\bibitem[Daouia et~al.(2023)Daouia, Stupfler, and Usseglio-Carleve]{daouia2023inference}
A.~Daouia, G.~Stupfler, and A.~Usseglio-Carleve.
\newblock Inference for extremal regression with dependent heavy-tailed data.
\newblock \emph{The Annals of Statistics}, 51\penalty0 (5):\penalty0 2040--2066, 2023.

\bibitem[Doherty and Richter(2002)]{doherty2002moral}
N.~A. Doherty and A.~Richter.
\newblock Moral hazard, basis risk, and gap insurance.
\newblock \emph{Journal of Risk and Insurance}, 69\penalty0 (1):\penalty0 9--24, 2002.

\bibitem[Elabed et~al.(2013)Elabed, Bellemare, Carter, and Guirkinger]{elabed2013managing}
G.~Elabed, M.~F. Bellemare, M.~R. Carter, and C.~Guirkinger.
\newblock Managing basis risk with multiscale index insurance.
\newblock \emph{Agricultural Economics}, 44\penalty0 (4-5):\penalty0 419--431, 2013.

\bibitem[Eze et~al.(2020)Eze, Girma, Zenebe, Kourouma, and Zenebe]{eze2020exploring}
E.~Eze, A.~Girma, A.~Zenebe, J.~M. Kourouma, and G.~Zenebe.
\newblock Exploring the possibilities of remote yield estimation using crop water requirements for area yield index insurance in a data-scarce dryland.
\newblock \emph{Journal of Arid Environments}, 183:\penalty0 104261, 2020.

\bibitem[Farkas et~al.(2021)Farkas, Lopez, and Thomas]{farkas2021cyber}
S.~Farkas, O.~Lopez, and M.~Thomas.
\newblock Cyber claim analysis using generalized pareto regression trees with applications to insurance.
\newblock \emph{Insurance: Mathematics and Economics}, 98:\penalty0 92--105, 2021.

\bibitem[Farkas et~al.(2024)Farkas, Heranval, Lopez, and Thomas]{farkas2024generalized}
S.~Farkas, A.~Heranval, O.~Lopez, and M.~Thomas.
\newblock Generalized pareto regression trees for extreme event analysis.
\newblock \emph{Extremes}, pages 1--41, 2024.

\bibitem[Gatzert et~al.(2019)Gatzert, Pokutta, and Vogl]{gatzert2019convergence}
N.~Gatzert, S.~Pokutta, and N.~Vogl.
\newblock Convergence of capital and insurance markets: Consistent pricing of index-linked catastrophe loss instruments.
\newblock \emph{Journal of Risk and Insurance}, 86\penalty0 (1):\penalty0 39--72, 2019.

\bibitem[Gu et~al.(2023)Gu, Li, Zhang, and Liu]{GU2023106350}
Z.~Gu, Y.~Li, M.~Zhang, and Y.~Liu.
\newblock Modelling economic losses from earthquakes using regression forests: Application to parametric insurance.
\newblock \emph{Economic Modelling}, 125:\penalty0 106350, 2023.
\newblock ISSN 0264-9993.
\newblock \doi{https://doi.org/10.1016/j.econmod.2023.106350}.
\newblock URL \url{https://www.sciencedirect.com/science/article/pii/S0264999323001621}.

\bibitem[H{\"a}rdle and Cabrera(2010)]{hardle2010calibrating}
W.~K. H{\"a}rdle and B.~L. Cabrera.
\newblock Calibrating cat bonds for mexican earthquakes.
\newblock \emph{Journal of Risk and Insurance}, 77\penalty0 (3):\penalty0 625--650, 2010.

\bibitem[Huang et~al.(2019)Huang, Nychka, and Zhang]{huang2019estimating}
W.~K. Huang, D.~W. Nychka, and H.~Zhang.
\newblock Estimating precipitation extremes using the log-histospline.
\newblock \emph{Environmetrics}, 30\penalty0 (4):\penalty0 e2543, 2019.

\bibitem[Jensen et~al.(2016)Jensen, Barrett, and Mude]{jensen2016index}
N.~D. Jensen, C.~B. Barrett, and A.~G. Mude.
\newblock Index insurance quality and basis risk: Evidence from northern kenya.
\newblock \emph{American Journal of Agricultural Economics}, 98\penalty0 (5):\penalty0 1450--1469, 2016.

\bibitem[Jonas(2019)]{jonas2019pandemic}
O.~Jonas.
\newblock Pandemic bonds: designed to fail in ebola.
\newblock \emph{Nature}, 572\penalty0 (7769):\penalty0 285--286, 2019.

\bibitem[Lin and Kwon(2020)]{lin2020application}
X.~Lin and W.~J. Kwon.
\newblock Application of parametric insurance in principle-compliant and innovative ways.
\newblock \emph{Risk Management and Insurance Review}, 23\penalty0 (2):\penalty0 121--150, 2020.

\bibitem[Louaas and Picard(2021)]{louaas2021optimal}
A.~Louaas and P.~Picard.
\newblock Optimal insurance coverage of low-probability catastrophic risks.
\newblock \emph{The Geneva Risk and Insurance Review}, 46:\penalty0 61--88, 2021.

\bibitem[Mao and Ostaszewski(2021)]{mao2021optimal}
H.~Mao and K.~Ostaszewski.
\newblock Optimal claim settlement strategies under constraint of cap on claim loss.
\newblock \emph{Mathematics}, 9\penalty0 (24):\penalty0 3284, 2021.

\bibitem[Mastroeni et~al.(2022)Mastroeni, Mazzoccoli, and Naldi]{mastroeni2022pricing}
L.~Mastroeni, A.~Mazzoccoli, and M.~Naldi.
\newblock Pricing cat bonds for cloud service failures.
\newblock \emph{Journal of Risk and Financial Management}, 15\penalty0 (10):\penalty0 463, 2022.

\bibitem[Michel-Kerjan and Morlaye(2008)]{michel2008extreme}
E.~Michel-Kerjan and F.~Morlaye.
\newblock Extreme events, global warming, and insurance-linked securities: How to trigger the “tipping point”.
\newblock \emph{The Geneva Papers on Risk and Insurance-Issues and Practice}, 33:\penalty0 153--176, 2008.

\bibitem[Ozaki et~al.(2008)Ozaki, Goodwin, and Shirota]{ozaki2008parametric}
V.~A. Ozaki, B.~K. Goodwin, and R.~Shirota.
\newblock Parametric and nonparametric statistical modelling of crop yield: implications for pricing crop insurance contracts.
\newblock \emph{Applied Economics}, 40\penalty0 (9):\penalty0 1151--1164, 2008.

\bibitem[Prokopchuk et~al.(2020)Prokopchuk, Prokopchuk, Mentel, and Bilan]{prokopchuk2020parametric}
O.~Prokopchuk, I.~Prokopchuk, G.~Mentel, and Y.~Bilan.
\newblock Parametric insurance as innovative development factor of the agricultural sector of economy.
\newblock \emph{AGRIS on-line Papers in Economics and Informatics}, 12\penalty0 (3):\penalty0 69--86, 2020.

\bibitem[Rai et~al.(2024)Rai, Hoffman, Lahiri, Nychka, Sain, and Bandyopadhyay]{rai2024fast}
S.~Rai, A.~Hoffman, S.~Lahiri, D.~W. Nychka, S.~R. Sain, and S.~Bandyopadhyay.
\newblock Fast parameter estimation of generalized extreme value distribution using neural networks.
\newblock \emph{Environmetrics}, 35\penalty0 (3):\penalty0 e2845, 2024.

\bibitem[Securities(2009)]{securities2009handbook}
I.-L. Securities.
\newblock The handbook of insurance-linked securities.
\newblock 2009.

\bibitem[Tan and Zhang(2024)]{tan2024flexible}
K.~S. Tan and J.~Zhang.
\newblock Flexible weather index insurance design with penalized splines.
\newblock \emph{North American Actuarial Journal}, 28\penalty0 (1):\penalty0 1--26, 2024.

\bibitem[Van~der Vaart(2000)]{van2000asymptotic}
A.~W. Van~der Vaart.
\newblock \emph{Asymptotic statistics}, volume~3.
\newblock Cambridge university press, 2000.

\bibitem[Van Der~Vaart et~al.(1996)Van Der~Vaart, Wellner, van~der Vaart, and Wellner]{van1996weak}
A.~W. Van Der~Vaart, J.~A. Wellner, A.~W. van~der Vaart, and J.~A. Wellner.
\newblock \emph{Weak convergence and Empirical Processes}.
\newblock Springer, 1996.

\bibitem[Velthoen et~al.(2023)Velthoen, Dombry, Cai, and Engelke]{velthoen2023gradient}
J.~Velthoen, C.~Dombry, J.-J. Cai, and S.~Engelke.
\newblock Gradient boosting for extreme quantile regression.
\newblock \emph{Extremes}, 26\penalty0 (4):\penalty0 639--667, 2023.

\bibitem[Zhang(2024)]{zhang2024blended}
J.~Zhang.
\newblock Blended insurance scheme: A synergistic conventional-index insurance mixture.
\newblock \emph{Insurance: Mathematics and Economics}, 119:\penalty0 93--105, 2024.

\bibitem[Zhou et~al.(2010)Zhou, Wu, and Wu]{zhou2010optimal}
C.~Zhou, W.~Wu, and C.~Wu.
\newblock Optimal insurance in the presence of insurer’s loss limit.
\newblock \emph{Insurance: Mathematics and Economics}, 46\penalty0 (2):\penalty0 300--307, 2010.

\end{thebibliography}

\end{document}